%% file: main.tex
\documentclass[copyright, creativecommons]{eptcs}
\providecommand{\event}{QPL 2022} 
\usepackage[utf8]{inputenc}
\usepackage{tikz-cd}
\usepackage{tikzit}
\usepackage{amsmath}
\usepackage{amsthm}
\usepackage{amssymb}
\usepackage{braket}
\usepackage{appendix}
\usepackage{hyperref}
\usepackage{stmaryrd}
\usepackage{float} 
\usepackage{enumitem}
\setlist{noitemsep} 

\input{tikzpaper/sample.tikzstyles}
\input{zx.tikzstyles}
\newtheorem{theorem}{Theorem}[section]
\newtheorem{lemma}[theorem]{Lemma}
\newtheorem{proposition}[theorem]{Proposition}

\newtheorem{observation}[theorem]{Observation}
\theoremstyle{definition}
\newtheorem{definition}[theorem]{Definition}
\newtheorem{example}[theorem]{Example}
\newtheorem{remark}[theorem]{Remark}

\newcommand{\CC}{\mathbb{C}}
\newcommand{\ZZ}{\mathbb{Z}}

\newcommand{\symdiff}{\ensuremath{\bigtriangleup}}

\newcommand{\intf}[1]{\left\llbracket #1 \right\rrbracket} 

\let\sse=\subseteq
\def\vc#1#2{#1 _1\zd #1 _{#2}}
\def\zd{,\ldots,}
\let\ld=\lambda
\def\abs#1{\left| #1 \right|}
\renewcommand{\t}[1]{\ensuremath{^{\otimes #1}}}

\newcommand{\phase}[2]{\begin{tikzpicture}[baseline=-0.1cm]
	\begin{pgfonlayer}{nodelayer}
		\node [style=none] (0) at (-0.4, 0) {};
		\node [style=none] (2) at (0.4, 0) {};
		\node [style=#1 phase dot] (5) at (0, 0) {#2};
	\end{pgfonlayer}
	\begin{pgfonlayer}{edgelayer}
		\draw (0.center) to (5);
		\draw (5) to (2.center);
	\end{pgfonlayer}
\end{tikzpicture}}
\newcommand{\had}{\begin{tikzpicture}[baseline=-0.1cm]
	\begin{pgfonlayer}{nodelayer}
		\node [style=none] (0) at (-0.4, 0) {};
		\node [style=none] (2) at (0.4, 0) {};
		\node [style=H gate] (5) at (0, 0) {};
	\end{pgfonlayer}
	\begin{pgfonlayer}{edgelayer}
		\draw (0.center) to (5);
		\draw (5) to (2.center);
	\end{pgfonlayer}
\end{tikzpicture}}

\title{Complete Flow-Preserving Rewrite Rules for MBQC Patterns with Pauli Measurements}

\author{Tommy McElvanney
\institute{School of Computer Science\\
University of Birmingham}
\email{txm639@student.bham.ac.uk}
\and
Miriam Backens
\institute{School of Computer Science\\
University of Birmingham}
\email{m.backens@cs.bham.ac.uk}
}
\def\titlerunning{Complete Flow-Preserving Rewrite Rules for MBQC Patterns with Pauli Measurements}
\def\authorrunning{T. McElvanney \& M. Backens}
\begin{document}
\maketitle

\begin{abstract}
 In the one-way model of measurement-based quantum computation (MBQC), computation proceeds via measurements on some standard resource state. So-called flow conditions ensure that the overall computation is deterministic in a suitable sense, with Pauli flow being the most general of these. Existing work on rewriting MBQC patterns while preserving the existence of flow has focused on rewrites that reduce the number of qubits.

 In this work, we show that introducing new $Z$-measured qubits, connected to any subset of the existing qubits, preserves the existence of Pauli flow. Furthermore, we give a unique canonical form for stabilizer ZX-diagrams inspired by recent work of Hu \& Khesin \cite{Hu_2021}. We prove that any MBQC-like stabilizer ZX-diagram with Pauli flow can be rewritten into this canonical form using only rules which preserve the existence of Pauli flow, and that each of these rules can be reversed while also preserving the existence of Pauli flow. Hence we have complete graphical rewriting for MBQC-like stabilizer ZX-diagrams with Pauli flow.
\end{abstract}

\maketitle

\section{Introduction}

The one-way model of measurement-based quantum computation (MBQC) shows how to implement quantum computations by successive adaptive single-qubit measurements on a resource state \cite{Raussendorf_One-Way_2001}, largely without using any unitary operations.
This contrasts with the more commonly-used circuit model and has applications in server-client scenarios as well as for certain quantum error-correcting codes.

An MBQC computation is given as a \emph{pattern}, which specifies the resource state -- usually a graph state -- and a sequence of measurements of certain types \cite{danos_parsimonious_2005}.
As measurements are non-deterministic, future measurements need to be adapted depending on the outcomes of past measurements to obtain an overall deterministic computation.
Yet not every pattern can be implemented deterministically.
Sufficient (and in some cases necessary) criteria for determinism are given by the different kinds of \emph{flow}, which define a partial order on the measured qubits and give instructions for how to adapt the future computation if a measurement yields the undesired outcome \cite{Danos_2006,Browne_2007} (cf.\ Section~\ref{s:pauli-flow}).

In addition to the applications mentioned above, the flexible structure of MBQC patterns is also useful as a theoretical tool.
For example, translations between circuits and MBQC patterns have been used to trade off circuit depth versus qubit number \cite{broadbent_parallelizing_2009} or to reduce the number of $T$-gates in a Clifford+T circuit \cite{kissinger_reducing_2020}.
When translating an MBQC pattern (back) into a circuit, it is important that the pattern still have flow, as circuit extraction algorithms rely on flow \cite{Danos_2006,miyazaki_analysis_2015,Duncan_2020,Backens_2021}

This work uses the ZX-calculus, a graphical language for representing and reasoning about quantum computations, which is convenient for representing both quantum circuits and MBQC patterns, and for translating between the two.
ZX-calculus diagrams directly corresponding to MBQC-patterns are said to be in \emph{MBQC form}.
The ZX-calculus has various complete sets of rewrite rules, meaning any two diagrams that represent the same linear map can be transformed into each other entirely graphically \cite{Backens_2014,Jeandel_2018,Ng_2017}.
Yet these rewrite rules do not necessarily preserve the existence of a flow, nor even the MBQC-form structure.
Thus, circuit optimisation using MBQC and the ZX-calculus relies on proofs that certain diagram rewrites do preserve both \cite{Duncan_2020,Backens_2021}.
Work so far has focused on rewrite rules that maintain or reduce the number of qubits, which find direct application in T-count optimisation \cite{Duncan_2020}.
Nevertheless, it is sometimes desirable to increase the number of qubits in an MBQC pattern while preserving the existence of flow, such as for more involved optimisation strategies \cite{staudacher_optimization_2021} or for obfuscation.

In this paper, we begin investigating rewrite rules that preserve the existence of flow while increasing the number of qubits.
In particular, we prove that a rewrite rule that introduces a new $Z$-measured qubit preserves flow.
Most work on flow-preserving rewriting so far has been done in the context of \emph{generalised flow}, also known as \emph{gflow} \cite{Browne_2007}, in either its simple \cite{Duncan_2020} or extended version \cite{Backens_2021}.
Yet with the qubit introduction rule, the setting shifts to that of \emph{Pauli flow} \cite{Browne_2007,Simmons_2021} since preserving the interpretation of the diagram requires that the new qubit be measured in the Pauli-$Z$ basis.

We show that adding this one new rule to the known flow-preserving rewrite rules suffices to get completeness for MBQC-form diagrams within the stabilizer fragment of the ZX-calculus.
To achieve completeness, we introduce a new unique normal form for stabilizer ZX-calculus diagrams, which is close to the MBQC form.
This normal form is based on work by Hu and Khesin \cite{Hu_2021} using the stabilizer graph notation of Elliott, Eastin and Caves \cite{Elliott_2008}, like the original stabilizer ZX-calculus completeness result \cite{Backens_2014}.
As the proof by Hu and Khesin is somewhat difficult to follow, we give an alternative uniqueness proof using the language of affine spaces.

The remainder of this paper is structured as follows: in Section~\ref{s:preliminaries}, we introduce the ZX-calculus, measurement-based quantum computing, and existing flow-preserving rewrite rules.
Section~\ref{s:canonical} contains the new canonical form and its uniqueness proof.
Section~\ref{s:completeness} presents the new flow-preserving rewrite rule and the completeness proof for the stabilizer MBQC-form fragment.
The conclusions are in Section~\ref{s:conclusions}.

\section{Preliminaries}
\label{s:preliminaries}

In this section, we give an overview of the ZX-calculus and then use it to introduce measurement-based quantum computing.
We discuss the notion of flow that will be used in this paper and some existing rewrite rules which preserve the existence of this flow.

\subsection{The ZX-calculus}
The ZX-calculus is a diagrammatic language for reasoning about quantum computations. We will provide a short introduction here; for a more thorough overview, see \cite{VDW_2020, Coecke_2017}.

A ZX-diagram consists of \textit{spiders} and \textit{wires}. Diagrams are read from left to right: wires entering a diagram from the left are inputs while wires exiting the diagram on the right are outputs, like in the quantum circuit model. ZX-diagrams compose in
two distinct ways: \textit{horizontal composition}, which involves connecting the output wires of one diagram to the input wires of another, and \textit{vertical composition} (or the tensor product), which just involves drawing one diagram vertically above the other.
The linear map corresponding to a ZX-diagram $D$ is denoted by $\intf{D}$.

ZX-diagrams are generated by two families of spiders which may have any number of inputs or outputs, corresponding to the Z and X bases respectively. $Z$-spiders are drawn as green dots and $X$-spiders as red dots; with $m$ inputs, $n$ outputs, and using $(\cdot)\t{k}$ to denote a $k$-fold tensor power, we have:
\[
 \intf{\input{tikzpaper/greenspider.tikz}}
 = \ket{0}\t{n}\bra{0}\t{m} + e^{i\alpha}\ket{1}\t{n}\bra{1}\t{m} \qquad\qquad
 \intf{\input{tikzpaper/redspider.tikz}}
 = \ket{+}\t{n}\bra{+}\t{m} + e^{i\alpha}\ket{-}\t{n}\bra{-}\t{m}
\]

Spiders with exactly one input and output are unitary, in particular $\intf{\phase{Z}{$\alpha$}} = \ket{0}\bra{0} +e^{i\alpha}\ket{1}\bra{1} = Z_\alpha$ and $\intf{\phase{X}{$\alpha$}} = \ket{+}\bra{+} +e^{i\alpha}\ket{-}\bra{-} = X_\alpha$.

Two diagrams $D$ and $D'$ are said to be equivalent if $\intf{D} = z \intf{D'}$ for some non-zero complex number $z$. For the rest of the paper, whenever we write a diagram equality we will mean equality up to some global scalar in this way.
For treatments of the ZX-calculus which do not ignore scalars see \cite{Backens_2015} for the stabilizer fragment, \cite{Jeandel_2018} for the Clifford+T fragment and \cite{Jeandel_2018_2, Ng_2017} for the full ZX-calculus.

The Hadamard gate $H=\ket{+}\bra{0}+\ket{-}\bra{1} \cong Z_{\frac{\pi}{2}} \circ X_{\frac{\pi}{2}} \circ Z_{\frac{\pi}{2}}$ will be used throughout the paper (where $\cong$ denotes equality up to non-zero scalar factor).
It has two common syntactic sugars -- a yellow square, or a blue dotted line -- with the latter only used between spiders:
\begin{center}
\input{tikzpaper/Hadamard.tikz} \qquad\qquad \input{tikzpaper/dashedHadamard.tikz}
\end{center}

The ZX-calculus is equipped with a set of rewrite rules which can be used to transform a ZX-diagram into another diagram representing the same linear map.
As this paper focuses on stabilizer quantum mechanics, we give a rule set for the stabilizer ZX-calculus in Figure~\ref{fig:ZX-rules}.
Together with the definition of \had{}, this set of rewrite rules is complete: any two stabilizer ZX-diagrams which correspond (up to non-zero scalar factor) to the same linear map can be rewritten into one another using these rules \cite{Backens_2014}.

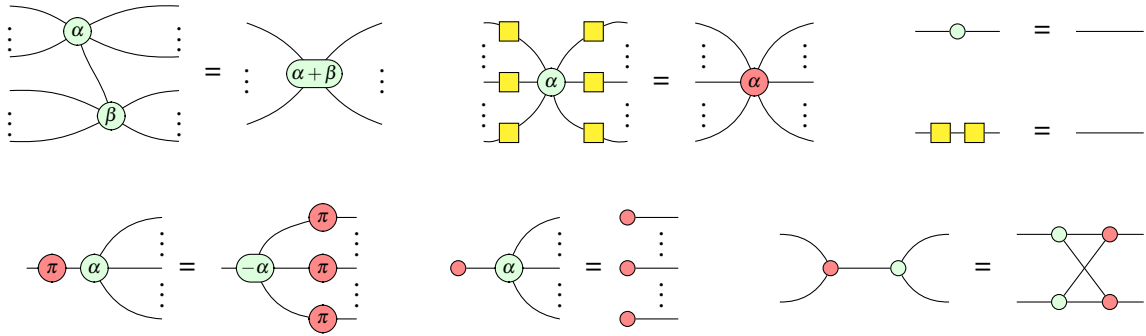
\begin{figure}
	\input{tikzpaper/ZX-rules.tikz}
	\caption{A complete set of rewrite rules for the scalar-free stabilizer ZX-calculus. Each rule also holds with the colours or the directions reversed.}
	\label{fig:ZX-rules}
\end{figure}

\subsection{Measurement-based Quantum computation}

Measurement-based Quantum computation (MBQC) is a particularly interesting model of quantum computation with no classical analogue. In MBQC, one first constructs a highly entangled resource state that can be independent of the specific computation that one wants to perform (only depending on the `size' of the computation) by preparing qubits in the $\ket{+}$ state and applying $CZ$-gates to certain pairs of qubits. The computation then proceeds by performing single qubit measurements in a specified order. MBQC is a universal model for quantum computation -- any computation can be performed by choosing an appropriate resource state and then performing a certain combination of measurements on said state.

Measurement-based computations are traditionally expressed as \textit{measurement patterns}, which use a sequence of commands to describe how the resource state is constructed and how the computation proceeds \cite{danos_parsimonious_2005}.
As the resource states are graph states, a graphical representation of MBQC protocols can be more intuitive; we shall therefore introduce MBQC with ZX-diagrams.

\begin{definition}[\cite{Duncan_2009}]
A \emph{graph state diagram} is a ZX-diagram where each vertex is a (phase-free) green spider,
each edge connecting spiders has a Hadamard gate on it, and
there is a single output wire incident on each vertex.  
A ZX-diagram is in \emph{graph state with local Clifford (GS-LC) form} if it is a graph state up to single qubit Clifford operators on the input and output wires.
It is in \emph{reduced GS-LC (rGS-LC) form} if those single-qubit Clifford operators are all in the set $\{\phase{Z}{$\frac{k\pi}{2}$}, \phase{Z}{$\pm\frac{\pi}{2}$}\!\phase{X}{$\frac{\pi}{2}$}\}$ for some $k\in \ZZ_4$ and if no two qubits with red phases in their vertex operator are connected to each other.
\end{definition}

\begin{definition}\cite[Definitions~2.18, 2.23]{Backens_2021}
\label{def:MBQC-form}
A ZX-diagram is in \emph{MBQC-form} if it consists of a
graph state diagram in which each vertex of the graph may furthermore be connected to an input (in addition to its output), and a measurement effect instead of its output.
A ZX-diagram is in \emph{MBQC+LC-form} if it is in MBQC-form up to single qubit Clifford operators on the input and output wires.
\end{definition}

MBQC restricts the allowed single-qubit measurements to three planes of the Bloch sphere: those spanned by the eigenstates of two Pauli matrices, called the XY, YZ and XZ planes. Each time a qubit $u$ is measured in a plane $\lambda(u)$ at an angle $\alpha$, one may obtain either the desired outcome, denoted $\bra{+_{\lambda(u),\alpha}}$, or the undesired outcome $\bra{-_{\lambda(u), \alpha}} = \bra{+_{\lambda(u), \alpha+\pi}}$.
Measurements where the angle is an integer multiple of $\frac{\pi}{2}$ are Pauli measurements; the corresponding measurement type is denoted by simply $X$, $Y$, or $Z$.
The ZX-diagram corresponding to each (desired) measurement outcome is given in Table~\ref{tab:MBQC-ZX}.
The structure of an MBQC protocol is formalised as follows.

\begin{table}
	\centering	
	\renewcommand{\arraystretch}{1.5}
	\begin{tabular}{c||c|c|c|c|c|c|c|c|c|}
		operator & $\bra{+_{XY, \alpha}}_i$ & $\bra{+_{XZ, \alpha}}_i$ & $\bra{+_{YZ, \alpha}}_i$ & $\bra{+_{X,0}}_i$ & $\bra{+_{Y,0}}_i$ & $\bra{+_{Z,0}}_i$ & $\bra{+_{X,\pi}}_i$ & $\bra{+_{Y,\pi}}_i$ & $\bra{+_{Z,\pi}}_i$ \\ \hline
		diagram &  \input{tikzpaper/XYmeasurement.tikz}  &  \input{tikzpaper/XZmeasurement.tikz} &  \input{tikzpaper/YZmeasurement.tikz} &  \input{tikzpaper/Xmeasurement.tikz}  &  \input{tikzpaper/Ymeasurement.tikz} &  \input{tikzpaper/Zmeasurement.tikz} &  \input{tikzpaper/Xmeasure2.tikz} &  \input{tikzpaper/Y-measure2.tikz} &  \input{tikzpaper/Zmeasure2.tikz}\\ 
	\end{tabular}
	\renewcommand{\arraystretch}{1}
	\caption{MBQC measurement effects in Dirac notation and their corresponding ZX-diagrams}
	\label{tab:MBQC-ZX}
\end{table}

\begin{definition} 	
A \emph{labelled open graph} is a tuple $\Gamma = (G, I, O, \lambda)$,
where $G=(V,E)$ is a simple undirected graph, $I\subseteq V$ is a set of input vertices, $O \subseteq V$ is a set of output vertices, and $\lambda: V\setminus O \to \{X, Y, Z, XY, XZ, YZ\}$ assigns a measurement plane or Pauli measurement to each non-output vertex.
\end{definition}

In this paper, we consider \emph{stabilizer MBQC diagrams}: MBQC-form diagrams where every non-output qubit has a Pauli measurement applied to it, i.e.\ where $\lambda: V \setminus O \to \{X, Y, Z\}$.

\subsection{Pauli flow}
\label{s:pauli-flow}

Measurement-based computations are inherently probabilistic because measurements are probabilistic.
Computations can be made deterministic overall (up to Pauli corrections on the outputs) by tracking which measurements result in undesired outcomes and then correcting for these by adapting future measurements.
A sufficient (and in some cases necessary) condition for this to be possible on a given labelled open graph is \emph{Pauli flow}.
In the following, $\mathcal{P}(S)$ denotes the powerset of a set $S$.

\begin{definition}[{\cite[Definition 5]{Browne_2007}}]
A labelled open graph $(G, I, O, \lambda)$ has Pauli flow if there exists a map $p: V\setminus O \to \mathcal{P}(V\setminus I)$ and a partial order $\prec$ over V such that for all $u \in V\setminus O$,
\begin{enumerate}
    \item if $v\in p(u)$, $v \not = u$ and $\lambda(v)\not \in \{X,Y\}$, then $u \prec v$. \label{P1}
    \item if $v\in \mathrm{Odd}_G(p(u))$, $v \not = u$ and $\lambda(v)\not \in \{Y, Z\}$, then $u \prec v$. \label{P2}
    \item if $\lnot (u \prec v)$ and $\lambda(v) = Y$, then $v \in p(u) \Longleftrightarrow v \in \mathrm{Odd}_G(p(u))$. \label{P3}
    \item if $\lambda(u) = XY$, then $u\not \in p(u)$ and $u \in \mathrm{Odd}_G(p(u))$. \label{P4}
    \item if $\lambda(u) = XZ$, then $u \in p(u)$ and $u \in \mathrm{Odd}_G(p(u))$. \label{P5}
    \item if $\lambda(u) = YZ$, then $u\in p(u)$ and $u \not \in \mathrm{Odd}_G(p(u))$.\label{P6}
    \item if $\lambda(u) = X$, then $u \in \mathrm{Odd}_G(p(u))$. \label{P7}
    \item if $\lambda(u) = Z$, then $u \in p(u)$. \label{P8}
    \item if $\lambda(u) = Y$ then either $u \in p(u)$ and $u\not \in \mathrm{Odd}_G(p(u))$ or $u \not \in p(u)$ and $u \in \mathrm{Odd}_G(p(u))$. \label{P9}
\end{enumerate}
\end{definition}

Here, the partial order restricts the time order in which the qubits need to be measured.
The set $p(u)$ denotes qubits that are modified by Pauli-$X$ to compensate for an undesired measurement outcome on $u$, $\mathrm{Odd}_G(p(u))$ denotes the set of vertices that are modified by Pauli-$Z$.

Pauli flow is a sufficient condition for strong, stepwise and uniform determinism: this means all branches of the computation should implement the same linear operator up to a phase, any interval of the computation should be deterministic on its own, and the computation should be deterministic for all choices of measurement angles that satisfy $\ld$ \cite[p.~5]{Browne_2007}.
Pauli flow (and related flow conditions) are particularly interesting from a ZX-calculus perspective as there are polynomial-time algorithms for extracting circuits from MBQC-form ZX-diagrams with flow \cite{Duncan_2020,Backens_2021,Simmons_2021}, while circuit extraction from general ZX-diagrams is \#P-hard \cite{deBeaudrap_2022}.

\subsection{Existing flow-preserving rewrite rules}
\label{s:existing-rules}

The basic ZX-calculus rewrite rules in Figure~\ref{fig:ZX-rules} do not generally preserve even the MBQC-form structure of a ZX-calculus diagram.
Yet there are some more complex derived rewrite rules that are known to preserve both the MBQC-form structure and the existence of a flow.
These rules were previously considered in the context of gflow \cite{Duncan_2020} and extended gflow \cite{Backens_2021}; the Pauli-flow preservation proofs are due to \cite{Simmons_2021}.
The simplest of these rules is $Z$-deletion:

\begin{lemma}[{\cite[Lemma D.6]{Simmons_2021}}]
\label{lem:Z-delete}
Deleting a $Z$-measured vertex preserves the existence of Pauli flow.
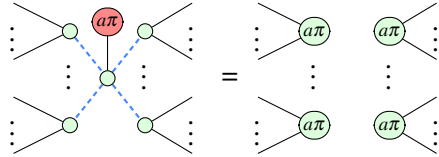
\begin{center}
    \input{tikzpaper/Z-deletenew.tikz}
\end{center}
\end{lemma}

Other rewrite rules are based around quantum generalisations of two graph-theoretic operations.

\begin{definition}
Let $G=(V,E)$ be a graph and $u\in V$. The \emph{local complementation of $G$ about $u$} is the operation which maps $G$ to $G\star u := (V, E\symdiff \{(b,c) | (b,u),(c,u) \in E \text{ and } b\not = c\})$, where $\symdiff$ is the symmetric difference operator given by $A\symdiff B = (A\cup B)\setminus(A\cap B)$.
 The \emph{pivot of $G$ about the edge $(u,v)$} is the operation mapping $G$ to the graph $G \wedge uv := G\star u \star v \star u$.
\end{definition}

Local complementation keeps the vertices of the graph the same but toggles some edges: for each pair of neighbours of $u$, i.e. $v, v' \in N_G(u)$, there is an edge connecting $v$ and $v'$ in $G\star u$ if and only if there is no edge connecting $v$ and $v'$ in $G$.
Pivoting is a series of three local complementations, but has some special properties which make it worth distinguishing.
It interchanges the vertices $u$ and $v$ and complements (or `toggles') the connectivity between the following three subsets of vertices \cite[Section~8]{Bouchet_1988}:

\begin{itemize}
	\item $N_G(u)\setminus (\{v\}\cup N_G(v))$, the neighbours of $u$ that are neither neighbours of $v$ nor $v$ itself.
	\item $N_G(v)\setminus (\{u\} \cup N_G(u))$, the neighbours of $v$ that are neither neighbours of $u$ nor $u$ itself.
	\item $N_G(u)\cap N_G(v)$, the common neighbours of $u$ and $v$.
\end{itemize}
From the above characterisation we see that pivoting is symmetric, i.e.\ $G\wedge uv = G\wedge vu$.

Both local complementation and pivoting give rise to operations on MBQC-form diagrams which preserve the MBQC form as well as the existence of Pauli flow (after some simple merging of single-qubit Cliffords into measurement effects, cf.~\cite[Section~4.2]{Backens_2021}).
We illustrate the operations with examples as they are difficult to express in ZX-calculus in generality.

\begin{lemma}[{\cite[Lemma D.12]{Simmons_2021}}]
 \label{lem:local-complementation}
 A local complementation about a vertex $u$ preserves the existence of Pauli flow.
 \begin{center}
      \input{tikzpaper/LC1.tikz}  \input{tikzpaper/LC2.tikz}
 \end{center}
\end{lemma}

\begin{lemma}[{\cite[Lemma D.21]{Simmons_2021}}]
 \label{lem:pivot}
 A pivot about an edge $(u,v)$ preserves the existence of Pauli flow.
\begin{center}
    \input{tikzpaper/pivot1.tikz} \input{tikzpaper/pivot2.tikz}
\end{center}
\end{lemma}

\begin{observation} \label{obs:inverses}
	Lemmas~\ref{lem:local-complementation} and \ref{lem:pivot} provide their own inverses since four successive local complementations about the same vertex, or two successive pivots about the same edge, leave the diagram invariant.
	Two successive local complementations correspond to the $\pi$-copy rule.
\end{observation}
While the inverse of the $Z$-deletion rule of Lemma~\ref{lem:Z-delete} straightforwardly preserves the MBQC-form, it is not obvious that it also preserves the existence of Pauli flow.
In Section~\ref{s:Z-insert}, we will prove that this is indeed the case.

\section{A canonical form for stabilizer state diagrams}
\label{s:canonical}

Stabilizer state diagrams in the ZX-calculus have a pseudo-normal form: the rGS-LC form, which arises from the representation of a stabilizer state in terms of a graph state and local Clifford operators \cite{Backens_2014}.

Here, we propose a new pseudo-normal form, based on the representation of a stabilizer state in terms of its affine support and a phase polynomial \cite{amy_polynomial-time_2014}.
Like the rGS-LC form, this is closely related to the stabilizer graphs of Elliott et al.~\cite{Elliott_2008} but it translates them into the ZX-calculus differently.
The new normal form allows (and in most cases requires) both green and red spiders, meaning it is not strictly `graph-like'.

Based on a recent proposal by Hu and Khesin \cite{Hu_2021}, we then show how to make this new pseudo-normal form unique, yielding a canonical form for stabilizer state diagrams in the ZX-calculus\footnote{At QCTIP 2022, we learned that an analogous result was independently derived by John van de Wetering~\cite{wetering_2022}.}.
In the process, we simplify the uniqueness proof of Hu and Khesin by making use of formalisms and results from the literature about holant problems.

We first prove some lemmas about the algebraic representation of stabilizer states which will be useful in proving uniqueness of the canonical form.
Next we introduce to the new pseudo-normal `phase polynomial form' and show how it corresponds to stabilizer states in phase-polynomial representation.
Finally, we define the canonical form, prove its uniqueness, and give an algorithm for rewriting diagrams into canonical form.
Throughout this section, diagrams contain red spiders and thus are not in MBQC+LC-form; yet by colour changing all of the red vertices and unfusing phases these can straightforwardly be transformed into MBQC+LC-form diagrams.

\subsection{Stabilizer states in terms of affine support and phase polynomial}
\label{s:afffine-phase}

It has long been known \cite{Dehaene_2003,VDN_2010} that an $n$-qubit stabilizer state can be written (up to normalisation) as
	\begin{equation}\label{eq:stabiliser_state}
 \sum_{x\in A} i^{l(x)} (-1)^{q(x)} \ket{x},
\end{equation}
where $A$ is an affine subspace of $\ZZ_2^n$, $l(x)= \sum_j d_j x_j$ for some fixed $d_j\in\ZZ_2$ is a linear function computed modulo 2, and $q(x) = \sum_{j<k} c_{jk}x_jx_k + \sum_j c_j x_j$ for some fixed $c_{jk},c_j\in\ZZ_2$ is a quadratic function.
The functions $l$ and $q$ together form a phase polynomial for the state, while $A$ determines the support.

Assuming $\dim(A)=n-m$, the elements of the affine space $A$ are the solutions to a set of linear equations $Rx=b$, where $R$ is an $m\times n$ binary matrix of rank $m$ (with $0\leq m\leq n$) and $b\in\ZZ_2^m$.
Each component of $x$ is considered a variable.
With respect to this linear system, the variables $\vc{x}{n}$ can be partitioned (not generally uniquely) into a set of $(n-m)$ \emph{free} variables and a set of $m$ \emph{dependent} variables such that every assignment of values to the free variables induces exactly one assignment of values to the dependent variables which satisfies all the linear equations.
This follows from a standard process of solving the system of linear equations, which also yields a linear equation in terms of the free variables for each dependent variable.
In the following, we will denote the set of indices by $[n] :=\{1,2\zd n\}$ and the free variables by a subset $F\sse [n]$ of the indices, and write the dependent variables as $x_j = a_j \oplus \bigoplus_{k\in F} a_{jk}x_k$, where $a_j,a_{jk}\in\ZZ_2$ and the sum is modulo 2.
If $a_{jk}=1$, we say the variable $x_j$ depends on $x_k$.

It will be useful to give a canonical choice of free variables, this is inspired by Hu and Khesin's normal form for stabilizer states \cite{Hu_2021}, and will lead us to an analogous normal form for stabilizer diagrams.

\begin{definition}\label{def:canonical-free}
 We call the result of the following procedure the canonical set of free variables.
 Start with $x_1$ and consider the variables in ascending order.
 For each $j$, if the value of $x_j$ is fixed by the requirement to satisfy $Rx=b$ given values for all free variables among $\vc{x}{j-1}$ then we say that $x_j$ is dependent. Otherwise we say that $x_j$ is free.
\end{definition}

\begin{lemma}\label{lem:canonical-variables-unique}
 Given an affine space $A$, the canonical set $F$ is the unique set of free variables with the following property: if $x_j$ depends on the free variable $x_k$, then $k<j$.
\end{lemma}
\begin{proof}
 Let $F'$ be another set of free variables for $A$ which also has the property that if $x_j$ is a dependent variable and depends on the free variable $x_k$, then $k<j$.
 In other words, for each $j\in [n]\setminus F'$, there is an equation $x_j = a_j + \sum_{k<j} a_{jk} x_k$, where furthermore $a_{jk}=0$ if $k\notin F'$.
 
 Now suppose for a contradiction that $F\neq F'$.
 The two sets must have the same size $\abs{F}=\abs{F'}=\dim(A)$.
 Thus, there must be a smallest element $j\in F$ such that $j\notin F'$.
 Then $F'$ induces an equation
 \begin{equation}\label{eq:dependent-var-linear-dec}
  x_j = a_j \oplus \bigoplus_{k\in F', \; k<j} a_{jk} x_k.
 \end{equation}
 Suppose $a_{jk}=1$ only if $k\in F$.
 Then the value of $x_j$ is fixed by the free variables of lower index in $F$, so $j$ should not be free according to Definition~\ref{def:canonical-free}, a contradiction.
 
 Otherwise, there exists some $k'\notin F$ such that $a_{jk'}=1$.
 But then by the definition of $F$, there exists some equation $x_{k'} = b_{k'} \oplus \bigoplus_{\ell\in F, \; \ell<k'} b_{k'\ell}$.
 Thus we can substitute for $x_{k'}$ in \eqref{eq:dependent-var-linear-dec} while preserving the property that $x_j$ only depends on variables of lower index.
 The process eliminates one variable which is not in $F$ from the decomposition and does not introduce any new variables which are not in $F$.
 Hence repeated application will terminate, at which point we have an equation that fixes $x_j$ from only variables in $F$ of index less than $j$.
 Again, this means $j$ should not be in $F$, a contradiction.

 Hence we must have $F=F'$.
\end{proof}

As pointed out in the holant literature, it is possible to express the functions $l$ and $q$ solely in terms of the free variables, while keeping their other properties the same \cite[Definition~8]{cai_dichotomy_2018}.
We give a proof in Appendix~\ref{a:canonical-form-proofs} for completeness.

\begin{lemma}\label{lem:phase-poly-free-only}
 Suppose $F$ denotes a set of free variables for the affine space $A$, and $\ket{\psi}$ is some stabilizer state with support on $A$.
 Then there exists a linear function $l$ and a quadratic function $q$, both depending only on the free variables, as well as a scalar $\ld\in\CC\setminus\{0\}$, such that:
 \[
  \ket{\psi} = \ld \sum_{x\in A} i^{l(x)} (-1)^{q(x)} \ket{x}.
 \]
\end{lemma}

There are generally multiple ways of expressing the same state in the form of \eqref{eq:stabiliser_state}.
Yet if we pick a set of free variables $F$ and require $l$ and $q$ to depend only on free variables, the representation becomes unique.
Moreover, we can even give a unique representation in terms of a phase polynomial (evaluated modulo 4, rather than 2).
Again, the proof is in Appendix~\ref{a:canonical-form-proofs}.

\begin{lemma}\label{lem:true-phase-poly}
 Given an $n$-qubit stabilizer state $\ket{\psi}$ and a set $F\sse [n]$, there exists a unique polynomial $p(x) = \sum_{j\in F} r_j x_j + 2\sum_{j,k\in F,\; j<k} s_{jk} x_j x_k$ with $r_j\in\ZZ_4$ and $s_{jk}\in \ZZ_2$ and scalar $\ld\in\CC\setminus\{0\}$ such that $\ket{\psi} = \ld\sum_{x\in A} i^{p(x)} \ket{x}$.
\end{lemma}

\subsection{A new pseudo-normal form related to phase polynomials}

In the rGS-LC form for stabilizer state diagrams, local Clifford operators on the graph state are expressed in terms of green and red spiders.
Alternatively, it is also possible to express local Clifford operators in terms of green spiders and Hadamards (and this is what is done in the stabilizer graph formalism of \cite{Elliott_2008}).
In ZX-terms, this means the allowed local Clifford operators are \phase{Z}{$\frac{k\pi}{2}$} and \phase{Z}{$a\pi$}\!\had, where $k\in\ZZ_4$ and $a\in\ZZ_2$.
As for red nodes in rGS-LC diagrams, qubits whose local Clifford operator contains an $H$ are not allowed to be connected to each other; therefore we can `push' the Hadamards through and get the following pseudo-normal form.
It is possible to convert between the two kinds of local Clifford operators via local complementations on the qubits that have red nodes or Hadamards.

\begin{definition}\label{def:phase-poly-form}
 A stabilizer ZX-calculus diagram is in phase-polynomial form if the following hold:
 \begin{itemize}
  \item Each dangling edge is connected to a unique red or green spider.
  \item Red spiders have phases that are 0 or $\pi$.
  \item Green spiders have phases that are integer multiples of $\pi/2$.
  \item There may be edges connecting spiders of different colours.
  \item Furthermore, green spiders may be connected to other green spiders via Hadamard nodes.
 \end{itemize}
\end{definition}

\begin{observation}\label{obs:rGS-LC_to_phasepoly}
An rGS-LC diagram can be brought into phase-polynomial form via the following process. First, apply local complementations to all qubits that have red nodes in their local Cliffords. This maps $\phase{Z}{$\frac{\pi}{2}$}\!\phase{X}{$\frac{\pi}{2}$}$ to $\had$ and $\phase{Z}{-$\frac{\pi}{2}$}\!\phase{X}{$\frac{\pi}{2}$}$ to $\had\!\phase{X}{$\pi$}$. Then, change the colour of all spiders which now have Hadamards as part of their vertex operators and merge adjacent spiders of the same colour.
\end{observation}

\begin{example}\label{ex:rGS-LCtoPhasePoly}
 Applying this procedure to the rGS-LC diagram on the left yields the phase polynomial-form diagram in the middle.
 Colour-changing each red spider and unfusing the phases leads to an equivalent GS-LC form diagram which we will say is in phase-polynomial form up to colour changing the spiders with Hadamard gates in their vertex operators.
  \begin{center}
    \input{tikzpaper/rGS-LC_new.tikz}\qquad\qquad\qquad
     \input{tikzpaper/nf_examplenew.tikz}\qquad\qquad\qquad \input{tikzpaper/GSnf_examplenew.tikz}
 \end{center}
\end{example}

Diagrams in phase-polynomial form correspond directly to pairs of a state and a set of free variables for the underlying affine support.
Appendix~\ref{a:bijection-example} contains an example illustrating this correspondence.

\begin{lemma}\label{lem:phase-poly-bijection}
 Ignoring scaling, there is a bijection between phase-polynomial form diagrams and pairs $\left(\ket{\psi}, \; F\right)$, where $\ket{\psi}$ is an $n$-qubit stabilizer state and $F\sse[n]$ indicates a set of free variables for the affine space $A$ which is the support of $\ket{\psi}$.
\end{lemma}
\begin{proof}
 By Lemma~\ref{lem:true-phase-poly}, there exists a unique function $p(x) = \sum_{j\in F} r_j x_j + 2\sum_{j,k\in F,\; j<k} s_{jk} x_j x_k$ with $r_j\in\ZZ_4$ and $s_{jk}\in\ZZ_2$ such that $\ket{\psi} \cong \sum_{x\in A} i^{p(x)} \ket{x}$.
 To construct a diagram from a state and a set of free variables from this, proceed as follows:
 \begin{itemize}
  \item For each dependent variable $x_k$ with $k\in [n]\setminus F$, find the unique linear expression $x_k = a_k \oplus \bigoplus_{j\in F} a_{kj} x_j$ which satisfies the defining linear equations $Rx=b$ of the affine space $A$.
  \item For each $j\in F$, place a green spider with an output wire. The phase of this spider is $r_j \frac{\pi}{2}$.
  \item For each $k\in [n]\setminus F$, place a red spider with an output wire. The phase of this spider is $a_j\pi$.
  \item Draw a (plain) edge connecting the green spider $j$ to the red spider $k$ whenever $a_{kj}=1$.
  \item Draw a Hadamard edge connecting the green spiders $j$ and $j'$ whenever $s_{jj'}=1$.
 \end{itemize}
 Conversely, given a diagram in phase-polynomial form, construct the corresponding state as below:
 \begin{itemize}
  \item The set $F$ of free variables consists of the indices of the green spiders.
  \item The affine space $A$ is defined by the set of equations $\left\{x_j = a_j \oplus \bigoplus_{k\in N(j)} x_k \right\}_{j\in [n]\setminus F}$,
   where $a_j=0$ if the phase of the red spider with index $j$ is 0, and 1 otherwise.
  \item For each $j\in F$ such that the phase of the green spider $j$ is $\alpha_j$, define $r_j$ to be the value in $\ZZ_4$ that is equivalent to $\frac{2\alpha_j}{\pi}\bmod 4$.
  \item For each $j,k\in F$ with $j<k$, define $s_{jk}=1$ if there exists a Hadamard edge between spiders $j$ and $k$, and $s_{jk}=0$ otherwise.
 \end{itemize}
  Let $p(x) := \sum_{j\in F} r_j x_j + 2\sum_{j,k\in F,\; j<k} s_{jk} x_j x_k$, then the desired state is $\sum_{x\in A} i^{p(x)}\ket{x}$.
  The two procedures are inverses of each other (noting that $\frac{3\pi}{2}\equiv-\frac{\pi}{2}\bmod 2\pi$).
  
  Suppose $D$ is the ZX-diagram corresponding to some stabilizer state $\ket{\psi}$ according to the above translation.
  Then it is straightforward to see that the support of $\intf{D}$ and the support of $\ket{\psi}$ are equal.
  Thus, by phase-polynomial techniques, it is quick to check that $\intf{D}$ equals $\ket{\psi}$ up to scalar factor.
\end{proof}

\subsection{The canonical phase-polynomial diagram}

Using the bijection between phase-polynomial form diagrams and pairs of a state and a set of free variables, we can now define a unique canonical diagram for any stabilizer state.

\begin{definition}\label{def:canonical-diagram}
 Let $\ket{\psi}$ be a stabilizer state, then its canonical diagram is the one translated from $(\ket{\psi},F)$ by Lemma~\ref{lem:phase-poly-bijection}, where $F$ is the canonical set of free variables according to Definition~\ref{def:canonical-free}.
\end{definition}

Apart from the translation into our terminology, this differs from the normal form definition of Hu and Khesin~\cite{Hu_2021} only by reversing the order: we ask for free variables to come first whereas they put them last.
Our uniqueness proof, making use of the properties of the affine support of a stabilizer state is shorter and simpler than that in \cite{Hu_2021}.

\begin{theorem}
 The canonical form is unique.
\end{theorem}
\begin{proof}
 This follows from the uniqueness of the canonical set of free variables proved in Lemma~\ref{lem:canonical-variables-unique} and from the bijection between pairs consisting of a state and a set of free variables in Lemma~\ref{lem:phase-poly-bijection}.
\end{proof}

\begin{proposition} \label{prop:phasepoly-to-canonical}
Every phase-polynomial form diagram can be re-written into canonical form using only local complementation and pivoting.
\end{proposition}
\begin{proof}
Pick some order $<$ on the spiders, say from top to bottom.
We want each red spider to only be connected to spiders that appear earlier in $<$.
While this does not hold, repeat the following procedure:
\begin{enumerate}
	\item Let $d_k$ be the minimal red spider under $<$ such that there exists some green spider $f_j$ connected to $d_k$ with $d_k < f_j$.
	
	\item Let $f_h$ be the maximal green spider under $<$ such that $d_k$ is connected to $f_h$.
	
	\item If $f_h$ has a phase of $\pm \frac{\pi}{2}$, perform local complementation about $f_h$ and then about $d_k$. Otherwise, pivot about the edge connecting $f_h$ and $d_k$.	After applying either of these equivalence transformations, $f_h$ is now red and $d_k$ is now green and the diagram is still in phase-polynomial form.
	\item By maximality of $f_h$, we have that $f_h$ is only connected to green spiders $f_n$ with $f_n < f_h$. By minimality of $d_k$, we have that $d_k$ is only connected to red spiders $d_m$ with $d_k < d_m$.
\end{enumerate}
This procedure strictly reduces the number of connections between red spiders and green spiders that appear later in the order.
Hence repeating it will eventually terminate, transforming any phase-polynomial form diagram into canonical form.
\end{proof}

\begin{remark}
	The canonical form is unique only
	up to the choice of order on the qubits; different orders may yield
	different `canonical forms'. Thus the choice of order is arbitrary (but
	needs to happen in advance, independently of the diagram considered) -- we
	have chosen top-to-bottom for simplicity.
\end{remark}

\section{Completeness}
\label{s:completeness}

Having established a canonical form for stabilizer ZX-calculus diagrams, we now give the completeness proof.
This first requires proving that a new rewrite rule preserves the existence of Pauli flow: an inverse to the $Z$-deletion rule of Lemma~\ref{lem:Z-delete}.
While there has been a lot of previous research on rewrite rules which reduce the number of spiders while preserving flow conditions, rewrite rules which increase the number of spiders have not been studied beyond introducing new degree-2 vertices along input or output wires (e.g.\ \cite[Lemma~4.1]{Backens_2021}).

\subsection{Inserting new $Z$-measured qubits}
\label{s:Z-insert}

Inserting $Z$-measured qubits into MBQC+LC form diagram preserves the existence of Pauli flow.

\begin{proposition}\label{prop:Z-insertion}
	Let $G=(V,E, I, O, \lambda)$ be a labelled open graph with Pauli flow and let $W\subseteq V$ be some arbitrary subset of the vertices. Then $G' = (V', E', I, O, \lambda')$ has a Pauli flow, where $V' = V \cup \{x\}$, $E' = E \cup \{(x, w) \mid w\in W\}$ with $\ld'(v)=\ld(v)$ if $v\neq x$ and $\ld'(x)=Z$.
\end{proposition}
\begin{proof}
	Let $(p,\prec)$ be a Pauli flow for $G$ and define $p': V'\setminus O \to \mathcal{P}(V'\setminus I)$ by $p'(v):=p(v)$ if $v\neq x$ and $p'(x) := \{x\}$.
	For vertices from the original graph, measurement planes and correction sets remain the same while the only change to odd neighbourhoods is that $x$ may be added. Thus conditions \ref{P4}--\ref{P7} and~\ref{P9} remain trivially satisfied. Condition \ref{P8} holds for $x$ as $x \in p'(x)$, and for all other $Z$-measured vertices because $(p,\prec)$ is a Pauli flow.
	
	Let $\prec'$ be the transitive closure of ${\prec} \cup \{(x,v) | v \in N_{G'}(x)\}$.
	Then $\prec'$ is a partial order because $\prec	$ is a partial order and we only add successors for $x$.
	Now, condition \ref{P1} of Pauli flow is inherited from $(p,\prec)$ for all $u \in V\setminus O$ because $u \not \in p'(x)$. Condition \ref{P2} is satisfied for all $u \in V\setminus O$ because $\lambda(x)=Z$ and $(p,\prec)$ is a Pauli flow. Condition \ref{P3} is inherited because the new vertex has only successors.
\end{proof}

\subsection{Complete flow-preserving rewrite rules}

We are now able to assemble the main proof. In the following, we will say an MBQC+LC-form diagram has \emph{no interior spiders} if the MBQC-form part of the diagram (i.e.\ ignoring the local Cliffords) has no interior vertices ($V\setminus(I\cup O)= \emptyset$). Additionally, we say an MBQC+LC-form diagram has Pauli flow if its MBQC-form part has Pauli flow (analogous to gflow in \cite[Section~4.1]{Backens_2021}).

\begin{theorem}\label{thm:completeness}
Given two equivalent stabilizer MBQC+LC-form diagrams $D$ and $D'$ with Pauli flow and satisfying $\intf{D}\cong\intf{D'}$, there exists a sequence of rewrite rules -- each preserving the existence of Pauli flow and preserving the MBQC+LC-form -- transforming $D$ into $D'$.
\end{theorem}
\begin{proof}
We begin by deleting all $Z$-measured vertices from both diagrams, keeping track of which vertices we delete and their set of neighbours when deleted.
The resulting diagrams has Pauli flow by Lemma~\ref{lem:Z-delete}. After all $Z$-measured vertices are removed, the MBQC-form parts of the diagrams (ignoring the local Cliffords) only have $X$ and $Y$ measurements and are thus of the kind considered in \cite{Duncan_2020}.
Then, there exists a terminating procedure (consisting of a sequence of local complementations, pivots and $Z$-deletions) rewriting the two diagrams into MBQC+LC-form diagrams $N$ and $N'$ which contain no interior spiders \cite[Theorem 5.4]{Duncan_2020}. Since local complementation and pivoting also preserve the existence of Pauli flow (Lemmas \ref{lem:local-complementation} and \ref{lem:pivot}), $N$ and $N'$ will also have Pauli flow.

As only $X$ and $Y$ measurements remain, they can be spider-merged and unmerged through each qubit to become local Cliffords on the outputs, thus $N$ and $N'$ are equivalent to GS-LC form diagrams. By \cite[Theorem 13]{Backens_2014}, every GS-LC form diagram can be rewritten into rGS-LC form using a sequence of local complementations, thus this step preserves Pauli flow. By Observation~\ref{obs:rGS-LC_to_phasepoly}, we can then rewrite each diagram into phase polynomial form, again using only local complementations (along with some operations on the local Cliffords that do not alter the flow), thus preserving Pauli flow.
Finally, by Proposition~\ref{prop:phasepoly-to-canonical}, we can rewrite each diagram into canonical form\footnotemark[\value{footnote}].
The rewrite steps use only local complementations and pivoting, so they preserve Pauli flow.
The resulting diagrams are equivalent and the canonical form is unique, so we have found a sequence of local complementations, pivots and $Z$-deletions rewriting $D$ and $D'$ into the same canonical form diagram $C$.

\footnotetext{Up to map-state duality and colour changing vertices with Hadamard operators.}
	
By Observation~\ref{obs:inverses}, local complementation and pivot can be inverted. Furthermore, $Z$-insert is a Pauli-flow preserving inverse to $Z$-delete.
Thus the sequence of rewrites from $D'$ to $C$ can be inverted while still preserving Pauli-flow.
By rewriting $D$ to $C$, then rewriting $C$ to $D'$, we obtain a sequence of flow-preserving rewrite rules transforming $D$ into $D'$.
This completes the proof.
\end{proof}

\begin{example}
We shall give a short example of this rewrite procedure in action.
	Consider the following two MBQC+LC-form diagrams, which we will call $D$ and $D'$, and which satisfy $\intf{D}\cong\intf{D'}$ by (non-flow preserving) diagram simplification techniques.
	\begin{center}
		  \input{tikzpaper/MBQCdiagrams2.tikz}
	\end{center}
Using the procedure from the proof of Theorem~\ref{thm:completeness}, we first rewrite $D$ to phase polynomial form.
 Perform triple local complementations (i.e.\ `inverse local complementations') about both the left-most and right-most qubits in the MBQC-form part, then apply $Z$-deletion to these qubits. A local complementation about the top left qubit gives us the fourth diagram, which is in rGS-LC form and in fact is equivalent to the left-most diagram in Example~\ref{ex:rGS-LCtoPhasePoly} up to map-state duality. We then obtain the final diagram by following the procedure in Observation~\ref{obs:rGS-LC_to_phasepoly}; note that this diagram is already in canonical form (up to map-state duality and colour changing spiders with Hadamard gates in their vertex operators) assuming that the input qubits have lower indices than the output qubits.
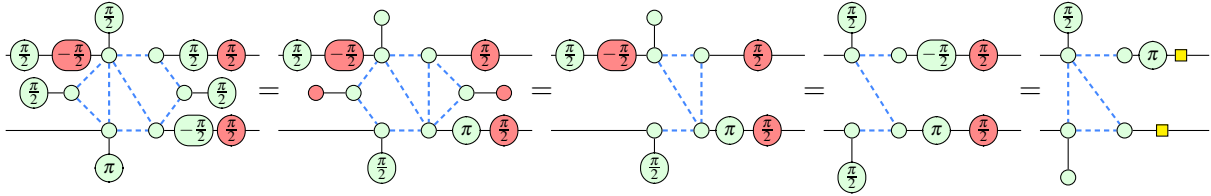
\begin{center}
	\input{tikzpaper/rewriteprocedure4.tikz}
\end{center}
For $D'$, we perform local complementation about the two interior qubits of the MBQC-form part (here we have done this about the top qubit first, then the bottom qubit), and $Z$-delete both qubits.
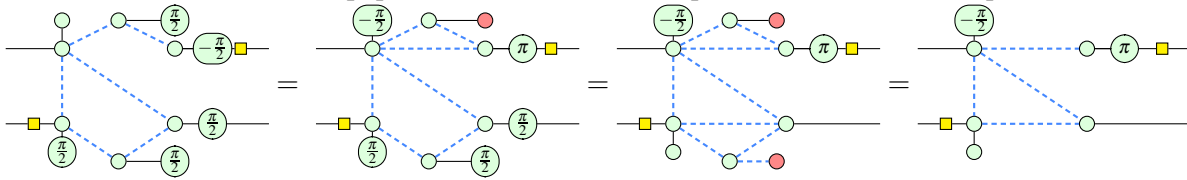
\begin{center}
	\input{tikzpaper/rewriteprocedure3.tikz}
\end{center}
This final diagram is already in phase polynomial form (up to map state duality and colour changing the spiders with Hadamard edges in their vertex operators) without us having to go through rGS-LC form. To rewrite this diagram into canonical form, all that remains is to pivot along the edge connecting the bottom left qubit to the bottom right qubit, giving the following diagram:

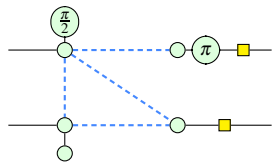
\begin{center}
	\input{tikzpaper/canonical3.tikz}
\end{center}
We have therefore rewritten $D$ and $D'$ into the same canonical form diagram. Every rule used to re-write $D$ and $D'$ to canonical form is invertible and the inverses preserve Pauli flow, giving us a sequence of flow preserving rewrite rules taking $D$ to $D'$.
\end{example}

\section{Conclusions}
\label{s:conclusions}

We have presented the first flow-preserving rewrite rule that increases the number of qubits in an MBQC-form ZX-diagram, and shown that this -- together with existing rewrite rules that preserve the MBQC form -- is complete for stabilizer MBQC-form diagrams.
The completeness proof goes via a new canonical form.
The result may find applications in obfuscation or in more involved optimisation protocols.

Yet that is only the beginning of the investigation of flow-preserving rewrite rules and in future work we will consider more extensive sets of rewrite rules and ZX-diagrams.
The recent proof that circuit extraction from general unitary ZX-diagrams is \#P-hard \cite{deBeaudrap_2022} means this line of research is particularly important, as it allows us to explore the only family of ZX-diagrams for which a polynomial-time circuit-extraction algorithm is currently known.

Pauli flow is known not to be necessary for deterministic implementability of MBQC patterns with all-Pauli measurements \cite{Browne_2007}; it would also be interesting to see how it can be extended and what flow-preserving rewriting would look like under the new conditions.

\paragraph*{Acknowledgements}

Thanks to Hex Miller-Bakewell for helpful comments on earlier notes about the phase-polynomial form.

\bibliographystyle{eptcs}
\bibliography{refs}

\appendix

\section{Algebraic proofs for the canonical form}
\label{a:canonical-form-proofs}

\begin{proof}[Proof of Lemma~\ref{lem:phase-poly-free-only}]
 In \eqref{eq:stabiliser_state}, the functions $l$ and $q$ are allowed to depend on all components of the bit string $x$, i.e.\ $l(x)= \bigoplus_j d_j x_j$ for some fixed $d_j\in\ZZ_2$ and $q(x) = \bigoplus_{j<k} c_{jk}x_jx_k \oplus \bigoplus_j c_j x_j$ for some fixed $c_{jk},c_j\in\ZZ_2$.
 
 Given the set of free variables $F$, solving the defining system of linear equations for $A$ yields linear equations $x_j = a_j\oplus\bigoplus_{k\in F} a_{jk}x_k$ for every $j\in[n]\setminus F$, where $a_j,a_{jk}\in\ZZ_2$.
 
 Now suppose $d_j\neq 0$ for some $j\notin F$.
 Then we can substitute
 \[
  l(x)= \bigoplus_{j\in [n]} d_j x_j = \left(\bigoplus_{j\in [n]\setminus\{s\}} d_j x_j\right) \oplus a_s \oplus \bigoplus_{t\in F} a_{st}x_t = a_s \oplus \bigoplus_{j\in [n]\setminus\{s\}} (d_j\oplus a_{sj}) x_j,
 \]
 where we define $a_{sj}=0$ if $j\notin F$.
 The $a_s$ is constant and the factor $i^{a_s}$ can be absorbed into the overall scalar $\ld$.
 Since $l$ is computed modulo 2, the new function satisfies the same properties as the original one but no longer depends on $x_s$.
 Furthermore, as $a_{sj}=0$ for all $j\notin F$, this process does not introduce any new dependencies on dependent variables.
 
 Therefore, the substitution process strictly decreases the number of dependent variables that $l$ depends on and successive applications will eventually yield a function that depends only on free variables.
 An analogous argument holds for $q$.
\end{proof}

\begin{lemma}\label{lem:phase-poly-uniqueness}
 Let $\ket{\psi}$ and $\ket{\phi}$ be two stabilizer states with the same support $A$, and let $F$ be a set of free variables for $A$.
 Suppose there exists $\ld,\mu\in\CC\setminus\{0\}$ such that
 \[
  \ket{\psi} = \ld \sum_{x\in A} i^{l(x)} (-1)^{q(x)} \ket{x} \qquad\text{and}\qquad
  \ket{\phi} = \mu \sum_{x\in A} i^{l'(x)} (-1)^{q'(x)} \ket{x}
 \]
 where for some $d_j,d_j',c_{jk},c_j,c_{jk}',c_j'\in\ZZ_2$,
 \begin{align*}
  l(x) &= \bigoplus_{j\in F} d_j x_j & q(x) &= \bigoplus_{j,k\in F, \; j<k} c_{jk} x_j x_k \oplus \bigoplus_j c_j x_j \\
  l'(x) &= \bigoplus_{j\in F} d_j' x_j & q'(x) &= \bigoplus_{j,k\in F,\; j<k} c_{jk}' x_j x_k \oplus \bigoplus_j c_j' x_j.
 \end{align*}
 Then $\ket{\psi}$ and $\ket{\phi}$ are linearly dependent if and only if for all $j,k\in F$ we have $d_j=d_j'$, $c_{jk}=c_{jk}'$, and $c_j=c_j'$.
\end{lemma}
\begin{proof}
 The `if' direction is straightforward: if $d_j=d_j'$, $c_{jk}=c_{jk}'$, and $c_j=c_j'$ for all $j,k\in F$, then $\mu\ket{\psi}=\ld\ket{\phi}$.

 For the `only if' direction, note that $l(x)=l'(x)=q(x)=q'(x)=0$ if all variables in $F$ are assigned 0, so by rescaling such that $\ld=\mu$, we get $\ket{\psi}=\ket{\phi}$ if and only if they are linearly dependent.
 
 By definition, each assignment of values to the free variables in $F$ induces one assignment of values to all the variables that is in $A$. 
 Suppose there exists a $j\in F$ such that $d_j\neq d_j'$, wlog assume $d_j=1$ and $d_j'=0$ (otherwise the argument is symmetric).
 Let $\xi$ be the bit string in $A$ that has every free variable set to 0 except the one with index $j$.
 Then $\braket{\xi\mid\psi}$ is imaginary while $\braket{\xi\mid\phi}$ is real, so since the two states have the same non-zero amplitude for the assignment induced by setting all free variables to 0, they cannot be linearly dependent.
 
 Similarly, suppose there exists $j\in F$ such that $c_j\neq c_j'$, then for the same $\xi$ we have $\braket{\xi\mid\psi}=-\braket{\xi\mid\phi}$, so again the two states cannot be linearly dependent.
 
 So without loss of generality, assume that $d_j=d_j'$ and $c_j=c_j'$ for all $j\in F$.
 Now suppose there are $j,k\in F$ such that $c_{jk}\neq c_{jk}'$.
 Let $\zeta$ be the bit string induced by the assignment where $x_j=x_k=1$ and all other free variables are 0.
 Then again, $\braket{\zeta\mid\psi}=-\braket{\zeta\mid\phi}$ so the two states cannot be linearly dependent.
 
 Therefore, linear dependence implies that for all $j,k\in F$ we have $d_j=d_j'$, $c_{jk}=c_{jk}'$, and $c_j=c_j'$.
\end{proof}

\begin{proof}[Proof of Lemma~\ref{lem:true-phase-poly}]
 Via Lemmas~\ref{lem:phase-poly-free-only} and ~\ref{lem:phase-poly-uniqueness}, we can uniquely write $\ket{\psi} = \ld \sum_{x\in A} i^{l(x)} (-1)^{q(x)} \ket{x}$, where $l(x)=\bigoplus_{j\in F} d_j x_j$ and $q(x)=\bigoplus_{j,k\in F,\; j<k} c_{jk}x_jx_k \oplus \bigoplus_j c_j x_j$ with all coefficients taking values in $\ZZ_2$.
 
 As $y\bmod 2 = y^2 \bmod 4$ for all $y\in\ZZ$, we have
\[
 \bigoplus_{j\in F} d_j x_j = \left(\sum_{j\in F} d_j x_j\right)^2 \bmod 4 =  \left(\sum_{j\in F} d_j x_j + 2\sum_{j,k\in F,\; j< k} d_jd_kx_jx_k\right) \bmod 4,
\]
where we have used the fact that $d_j,x_j\in\ZZ_2$ for all $j$ and hence $(d_jx_j)^2=d_jx_j$.
We can thus write
\[
 \sum_{x\in A} i^{l(x)} (-1)^{q(x)} \ket{x} = \sum_{x\in A} i^{p(x)} \ket{x},
\]
where
\begin{align*}
 p(x) &= \sum_{j\in F} d_j x_j + 2\left(\sum_{j,k\in F,\; j<k} c_{jk}x_jx_k + \sum_j c_j x_j\right) + 2\sum_{j,k\in F,\; j<k} d_jd_k x_jx_k \\
 &= \sum_{j\in F} (d_j+2c_j) x_j + 2\sum_{j,k\in F,\; j<k} (c_{jk} + d_j d_k) x_j x_k
\end{align*}
Now, $r_j:= (d_j+2c_j)\in\ZZ_4$.
The coefficient $s_{jk} := c_{jk}+d_jd_k$ could take value 2, but as $p$ is in the exponent of $i$ and $s_{jk}$ is multiplied by 2, we may without loss of generality replace it with $s_{jk} := c_{jk}\oplus d_j d_k$ so that $s_{jk}\in\ZZ_2$.
 
 Conversely, we can find functions $l$ and $q$ from $p$ by setting $d_j := r_j \bmod 2$, $c_j := \frac{1}{2}(r_j-d_j)$, and $c_{jk}:= s_{jk}\oplus d_j d_k$.
 Thus, by uniqueness of $l$ and $q$, the phase polynomial expression is also unique.
\end{proof}

\section{An example illustrating Lemma~\ref{lem:phase-poly-bijection}}
\label{a:bijection-example}

Consider the following phase-polynomial form diagram from Example~\ref{ex:rGS-LCtoPhasePoly}, where we have numbered the qubits from top to bottom.

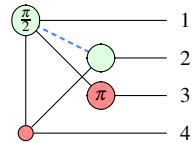
\begin{center}
	\input{tikzpaper/nf-example-numbered.tikz}
\end{center}

Following the procedure from Lemma~\ref{lem:phase-poly-bijection}, we construct the state corresponding to this diagram.
The state will be expressed as $\sum_{x\in A} i^{p(x)}\ket{x}$, where $p(x) = \sum_{j\in F} r_j x_j + 2\sum_{j,k\in F,\; j<k} s_{jk} x_j x_k$.
Here, $F$ is the set of free variables, $A$ is the affine space on which the state has support, and $p(x)$ is the phase polynomial with $r_j\in\ZZ_4$ and $s_{jk}\in\ZZ_2$ for all $j,k\in F$.
\begin{itemize}
 \item The set $F$ of free variables corresponding to this diagram is $F= \{x_1, x_2\}$ since qubits 1 and 2 are denoted by green spiders.
 
 \item The affine space $A$ is defined by the following set of equations arising from the red spiders:
 \begin{equation}\label{eq:example-affine}
  x_3 = 1 \oplus x_1 \qquad\qquad\qquad x_4 = x_1 \oplus x_2
 \end{equation}
 since qubit 3 has phase $\pi$ (giving the constant 1 on the right-hand side) and is connected to qubit 1, while qubit 4 has phase 0 and is connected to both 1 and 2.
 
 \item For the linear terms in the phase polynomial, we get that $r_1 = 1$ and $r_2 = 0$ as the phase of $x_1$ is $\frac{\pi}{2}$ and the phase of $x_2$ is $0$.
 
 \item For the quadratic terms in the phase polynomial, we have $s_{12}=1$ as there is a Hadamard edge connecting $x_1$ and $x_2$.
\end{itemize}
Combining these, the phase polynomial is $p(x) = x_1 + 2 x_1 x_2$.
The state corresponding to the diagram is therefore given by:
\begin{align*}
 \sum_{x\in A} i^{x_1 + 2 x_1 x_2}\ket{x}
 &= \sum_{x_1, x_2\in \ZZ_2} i^{x_1} (-1)^{x_1 x_2} \ket{x_1 x_2 (1\oplus x_1) (x_1\oplus x_2)} \\
 &= \ket{0010}+\ket{0111}+ i\ket{1001} -i\ket{1100}
\end{align*}
It is then quick to check that applying the procedure in Lemma~\ref{lem:phase-poly-bijection} for constructing a diagram from a state and a set of free variables gives back the original diagram.

Instead, we will show how to construct the diagram corresponding to the same state with a different set of free variables $F = \{x_2, x_3\}$.
To do this, we first rewrite the affine space and the phase polynomial in terms of the new free variables $x_2$ and $x_3$, and then apply the procedure for obtaining diagrams.

Choosing $x_3$ to be free instead of $x_1$, we rearrange the first equation of \eqref{eq:example-affine} and then substitute it into the second to get:
\begin{equation}\label{eq:affine-new}
 x_1 = 1\oplus x_3 \qquad\qquad\qquad x_4 = 1\oplus x_2 \oplus x_3
\end{equation}
Substituting into the phase polynomial yields $p(x) = (1\oplus x_3) + 2(1\oplus x_3)x_2$ where $\oplus$ denotes addition modulo 2.
Yet we want the phase polynomial to be computed modulo 4, since $i^4=1$.
Now, as $y\bmod 2 = y^2 \bmod 4$ for all $y\in\ZZ$, and $b^2 = b$ for all $b\in\ZZ_2$, this can be rewritten to:
\begin{align*}
 p(x) = (1\oplus x_3) + 2(1\oplus x_3)x_2
 = (1+x_3)^2 + 2(1+x_3)^2 x_2
 = 1 + 2 x_2 + 3 x_3 + 2 x_2 x_3 \pmod 4
\end{align*}
We thus have $r_2 =2$, $r_3=3$, and $s_{23}=1$.
The constant term in the phase polynomial is irrelevant since we are ignoring global scalars.
Up to scalar factor, the full state is
\[
 \sum_{x_2,x_3\in\ZZ_2} i^{2 x_2 + 3 x_3 + 2 x_2 x_3} \ket{(1\oplus x_3) x_2 x_3 (1\oplus x_2 \oplus x_3)}.
\]
To construct the diagram corresponding to this state and set of free variables:
\begin{itemize}
 \item We already have the equations for the dependent variables in terms of $F=\{x_2,x_3\}$ in \eqref{eq:affine-new}.
 
 \item Place a green spider with phase $r_2\frac{\pi}{2} = \pi$ for qubit 2 and a green spider with phase $r_3\frac{\pi}{2} = \frac{3\pi}{2}$ (or, equivalently, $-\frac{\pi}{2}$) for qubit 3.
 Each of the spiders is connected to one output wire.
 
 \item Place a red spider with phase $\pi$ for qubit 1 and a red spider with phase $\pi$ for qubit 4 since the equations for both $x_1$ and $x_4$ contain a constant term. Again, each of the spiders is connected to one output wire.
 
 \item Variable $x_1$ depends on $x_3$, so draw a plain wire between the spiders for qubits 1 and 3.
 Variable $x_4$ depends on both $x_2$ and $x_3$, so draw plain wires between the spiders for qubits 2 and 4, as well as between 3 and 4.
 
 \item As $s_{23} =1$, draw a Hadamard edge connecting the green spiders corresponding to $x_2$ and $x_3$.
\end{itemize}
This yields the following diagram:
\begin{center}
		\input{tikzpaper/example38.tikz}
\end{center}

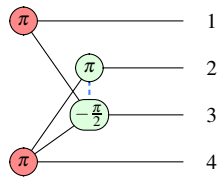
\end{document}

%% file: tikzpaper/sample.tikzstyles
\definecolor{zx_red}{RGB}{232, 165, 165}
\definecolor{zx_green}{RGB}{216, 248, 216}

\tikzstyle{gn}=[rectangle,rounded corners=0.8em,fill=zx_green,draw=black,
  line width=0.8 pt,inner sep=3pt,minimum width=1.5em,minimum height=1.5em]
\tikzstyle{rn}=[rectangle,rounded corners=0.8em,fill=zx_red,draw=black,
  line width=0.8 pt,inner sep=3pt,minimum width=1.5em,minimum height=1.5em]
\tikzstyle{medium box}=[fill=white, draw=black, shape=rectangle, minimum width=0.75cm, minimum height=1cm]
\tikzstyle{small box}=[fill=white, draw=black, shape=rectangle, minimum width=0.5cm, minimum height=0.66cm]
\tikzstyle{H gate}=[fill={rgb,255: red,255; green,242; blue,52}, draw=black, shape=rectangle, minimum width=0.25cm, minimum height=0.25cm]
\tikzstyle{blackdot}=[fill=black, draw=black, shape=circle]

\tikzstyle{Hadamard edge}=[-, dashed, dash pattern=on 2pt off 1.5pt, thick, draw={rgb,255: red,68; green,136; blue,255}]

%% file: zx.tikzstyles
\tikzstyle{box}=[shape=rectangle, text height=1.5ex, text depth=0.25ex, yshift=0.5mm, fill=white, draw=black, minimum height=5mm, yshift=-0.5mm, minimum width=5mm, font={\small}]
\tikzstyle{Z dot}=[inner sep=0mm, minimum size=2mm, shape=circle, draw=black, fill={rgb,255: red,221; green,255; blue,221}]
\tikzstyle{Z phase dot}=[minimum size=1.2em, font={\footnotesize\boldmath}, shape=rectangle, rounded corners=0.5em, inner sep=0.2em, outer sep=-0.2em, scale=0.8, tikzit shape=circle, draw=black, fill={rgb,255: red,221; green,255; blue,221}, tikzit draw=blue]
\tikzstyle{X dot}=[Z dot, shape=circle, draw=black, fill={rgb,255: red,255; green,136; blue,136}]
\tikzstyle{X phase dot}=[Z phase dot, tikzit shape=circle, tikzit draw=blue, fill={rgb,255: red,255; green,136; blue,136}, font={\footnotesize\boldmath}]
\tikzstyle{hadamard}=[fill=yellow, draw=black, shape=rectangle, inner sep=0.6mm, minimum height=1.5mm, minimum width=1.5mm]
\tikzstyle{vertex}=[inner sep=0mm, minimum size=1mm, shape=circle, draw=black, fill=black]
\tikzstyle{vertex set}=[inner sep=0mm, minimum size=1mm, shape=circle, draw=black, fill=white, font={\footnotesize\boldmath}]
\tikzstyle{target}=[inner sep=0mm, minimum size=3mm, shape=circle, draw=black]

\tikzstyle{hadamard edge}=[-, dashed, dash pattern=on 2pt off 1.5pt, thick, draw={rgb,255: red,68; green,136; blue,255}]
\tikzstyle{brace edge}=[-, tikzit draw=blue, decorate, decoration={brace,amplitude=1mm,raise=-1mm}]
\tikzstyle{diredge}=[->]
\tikzstyle{dashed edge}=[-, dashed, dash pattern=on 2pt off 0.5pt, draw=black]

%% file: tikzpaper/greenspider.tikz
\begin{tikzpicture}[baseline=-0.1cm, scale=0.5]
	\begin{pgfonlayer}{nodelayer}
		\node [style=Z phase dot] (0) at (0, 0) {$\alpha$};
		\node [style=none] (1) at (-1, 1) {};
		\node [style=none] (2) at (-1, 0.5) {};
		\node [style=none] (3) at (-1, -1) {};
		\node [style=none] (4) at (1, -1) {};
		\node [style=none] (5) at (1, 0.5) {};
		\node [style=none] (6) at (1, 1) {};
		\node [style=none] (7) at (-0.85, -0.15) {\vdots};
		\node [style=none] (8) at (0.85, -0.15) {\vdots};
	\end{pgfonlayer}
	\begin{pgfonlayer}{edgelayer}
		\draw (0) to (2.center);
		\draw [bend right=15] (0) to (1.center);
		\draw [bend left=15] (0) to (3.center);
		\draw [bend right=15] (0) to (4.center);
		\draw (0) to (5.center);
		\draw [bend left=15] (0) to (6.center);
	\end{pgfonlayer}
\end{tikzpicture}

%% file: tikzpaper/redspider.tikz
\begin{tikzpicture}[baseline=-0.1cm, scale=0.5]
	\begin{pgfonlayer}{nodelayer}
		\node [style=X phase dot] (0) at (0, 0) {$\alpha$};
		\node [style=none] (1) at (-1, 1) {};
		\node [style=none] (2) at (-1, 0.5) {};
		\node [style=none] (3) at (-1, -1) {};
		\node [style=none] (4) at (1, -1) {};
		\node [style=none] (5) at (1, 0.5) {};
		\node [style=none] (6) at (1, 1) {};
		\node [style=none] (7) at (-0.85, -0.15) {\vdots};
		\node [style=none] (8) at (0.85, -0.15) {\vdots};
	\end{pgfonlayer}
	\begin{pgfonlayer}{edgelayer}
		\draw (0) to (2.center);
		\draw [bend right=15] (0) to (1.center);
		\draw [bend left=15] (0) to (3.center);
		\draw [bend right=15] (0) to (4.center);
		\draw (0) to (5.center);
		\draw [bend left=15] (0) to (6.center);
	\end{pgfonlayer}
\end{tikzpicture}

%% file: tikzpaper/Hadamard.tikz
\begin{tikzpicture}[baseline=-0.1cm]
	\begin{pgfonlayer}{nodelayer}
		\node [style=none] (0) at (-2.75, 0) {};
		\node [style=none] (1) at (-1.75, 0) {};
		\node [style=none] (3) at (-0.75, 0) {};
		\node [style=none] (4) at (1.75, 0) {};
		\node [style=H gate] (5) at (-2.25, 0) {};
		\node [style=Z phase dot] (6) at (-0.25, 0) {$\frac{\pi}{2}$};
		\node [style=X phase dot] (7) at (0.5, 0) {$\frac{\pi}{2}$};
		\node [style=Z phase dot] (8) at (1.25, 0) {$\frac{\pi}{2}$};
		\node [style=none] (9) at (-1.25, 0) {$=$};
	\end{pgfonlayer}
	\begin{pgfonlayer}{edgelayer}
		\draw (0.center) to (5);
		\draw (5) to (1.center);
		\draw (3.center) to (6);
		\draw (6) to (7);
		\draw (7) to (8);
		\draw (8) to (4.center);
	\end{pgfonlayer}
\end{tikzpicture}

%% file: tikzpaper/dashedHadamard.tikz
\begin{tikzpicture}[baseline=-0.1cm]
	\begin{pgfonlayer}{nodelayer}
		\node [style=none] (0) at (-3, 0) {};
		\node [style=Z dot] (1) at (-2.5, 0) {};
		\node [style=H gate] (2) at (-1.75, 0) {};
		\node [style=Z dot] (3) at (-1, 0) {};
		\node [style=none] (4) at (-0.5, 0) {};
		\node [style=none] (5) at (0, 0) {$=$};
		\node [style=none] (6) at (0.5, 0) {};
		\node [style=Z dot] (7) at (1, 0) {};
		\node [style=Z dot] (8) at (2, 0) {};
		\node [style=none] (9) at (2.5, 0) {};
	\end{pgfonlayer}
	\begin{pgfonlayer}{edgelayer}
		\draw (0.center) to (1);
		\draw (1) to (2);
		\draw (2) to (3);
		\draw (3) to (4.center);
		\draw (6.center) to (7);
		\draw (8) to (9.center);
		\draw [style=Hadamard edge] (7) to (8);
	\end{pgfonlayer}
\end{tikzpicture}

%% file: tikzpaper/ZX-rules.tikz
\begin{tikzpicture}[scale=0.45]
	\begin{pgfonlayer}{nodelayer}
		\node [style=Z phase dot] (0) at (-17, 4.75) {$\alpha$};
		\node [style=Z phase dot] (1) at (-16, 2.25) {$\beta$};
		\node [style=none] (2) at (-19, 5.5) {};
		\node [style=none] (3) at (-19, 4) {};
		\node [style=none] (4) at (-19, 3) {};
		\node [style=none] (5) at (-19, 1.5) {};
		\node [style=none] (6) at (-14, 5.5) {};
		\node [style=none] (7) at (-14, 4) {};
		\node [style=none] (8) at (-14, 3) {};
		\node [style=none] (9) at (-14, 1.5) {};
		\node [style=none] (10) at (-19, 4.75) {$\vdots$};
		\node [style=none] (11) at (-19, 2.25) {$\vdots$};
		\node [style=none] (12) at (-14, 2.25) {$\vdots$};
		\node [style=none] (13) at (-14, 4.75) {$\vdots$};
		\node [style=none] (14) at (-13, 3.5) {=};
		\node [style=Z phase dot] (15) at (-10, 3.5) {$\alpha + \beta$};
		\node [style=none] (16) at (-12, 5) {};
		\node [style=none] (17) at (-12, 2) {};
		\node [style=none] (18) at (-8, 5) {};
		\node [style=none] (19) at (-8, 2) {};
		\node [style=none] (20) at (-12, 3.5) {$\vdots$};
		\node [style=none] (21) at (-8, 3.5) {$\vdots$};
		\node [style=Z phase dot] (23) at (-3, 3.25) {$\alpha$};
		\node [style=none] (24) at (-5, 5) {};
		\node [style=none] (25) at (-5, 3.25) {};
		\node [style=none] (26) at (-5, 1.5) {};
		\node [style=none] (27) at (-0.75, 5) {};
		\node [style=none] (28) at (-0.75, 3.25) {};
		\node [style=none] (29) at (-0.75, 1.5) {};
		\node [style=H gate] (30) at (-4.25, 4.75) {};
		\node [style=H gate] (31) at (-4.25, 3.25) {};
		\node [style=H gate] (32) at (-4.25, 1.75) {};
		\node [style=H gate] (33) at (-1.75, 1.75) {};
		\node [style=H gate] (34) at (-1.75, 3.25) {};
		\node [style=H gate] (35) at (-1.75, 4.75) {};
		\node [style=none] (36) at (0.25, 3.25) {=};
		\node [style=X phase dot] (38) at (3, 3.25) {$\alpha$};
		\node [style=none] (39) at (1.25, 3.25) {};
		\node [style=none] (40) at (1.25, 5) {};
		\node [style=none] (41) at (1.25, 1.5) {};
		\node [style=none] (42) at (4.75, 5) {};
		\node [style=none] (43) at (4.75, 3.25) {};
		\node [style=none] (44) at (4.75, 1.5) {};
		\node [style=none] (45) at (7.75, 4.75) {};
		\node [style=none] (46) at (10.25, 4.75) {};
		\node [style=none] (47) at (12.5, 4.75) {};
		\node [style=none] (48) at (14.5, 4.75) {};
		\node [style=none] (49) at (11.5, 4.75) {=};
		\node [style=Z dot] (51) at (9, 4.75) {};
		\node [style=none] (52) at (7.75, 1.75) {};
		\node [style=none] (53) at (10.25, 1.75) {};
		\node [style=none] (54) at (12.5, 1.75) {};
		\node [style=none] (55) at (14.5, 1.75) {};
		\node [style=none] (56) at (11.5, 1.75) {=};
		\node [style=H gate] (58) at (8.5, 1.75) {};
		\node [style=H gate] (59) at (9.5, 1.75) {};
		\node [style=none] (60) at (-18.5, -2.25) {};
		\node [style=X phase dot] (61) at (-17.75, -2.25) {$\pi$};
		\node [style=Z phase dot] (62) at (-16.5, -2.25) {$\alpha$};
		\node [style=none] (63) at (-14.5, -0.75) {};
		\node [style=none] (64) at (-14.5, -2.25) {};
		\node [style=none] (65) at (-14.5, -3.75) {};
		\node [style=none] (66) at (-14.5, -1.3) {$\vdots$};
		\node [style=none] (67) at (-14.5, -2.8) {$\vdots$};
		\node [style=none] (68) at (-12.75, -2.25) {};
		\node [style=Z phase dot] (70) at (-11.75, -2.25) {$-\alpha$};
		\node [style=none] (71) at (-8.75, -0.75) {};
		\node [style=none] (72) at (-8.75, -2.25) {};
		\node [style=none] (73) at (-8.75, -3.75) {};
		\node [style=none] (74) at (-8.75, -1.3) {$\vdots$};
		\node [style=none] (75) at (-8.75, -2.8) {$\vdots$};
		\node [style=none] (76) at (-13.75, -2.25) {=};
		\node [style=X phase dot] (77) at (-9.75, -0.75) {$\pi$};
		\node [style=X phase dot] (78) at (-9.75, -2.25) {$\pi$};
		\node [style=X phase dot] (79) at (-9.75, -3.75) {$\pi$};
		\node [style=X dot] (82) at (-5.75, -2.25) {};
		\node [style=Z phase dot] (83) at (-4.25, -2.25) {$\alpha$};
		\node [style=none] (84) at (-2.75, -0.75) {};
		\node [style=none] (85) at (-2.75, -2.25) {};
		\node [style=none] (86) at (-2.75, -3.75) {};
		\node [style=none] (87) at (-2.75, -1.3) {$\vdots$};
		\node [style=none] (88) at (-2.75, -2.8) {$\vdots$};
		\node [style=none] (89) at (-1.75, -2.25) {=};
		\node [style=X dot] (91) at (-0.75, -0.75) {};
		\node [style=X dot] (92) at (-0.75, -2.25) {};
		\node [style=X dot] (93) at (-0.75, -3.75) {};
		\node [style=none] (94) at (0.75, -0.75) {};
		\node [style=none] (95) at (0.75, -2.25) {};
		\node [style=none] (96) at (0.75, -3.75) {};
		\node [style=none] (97) at (3.75, -1.25) {};
		\node [style=none] (98) at (3.75, -3.25) {};
		\node [style=X dot] (99) at (5.25, -2.25) {};
		\node [style=Z dot] (100) at (7.25, -2.25) {};
		\node [style=none] (101) at (8.75, -1.25) {};
		\node [style=none] (102) at (8.75, -3.25) {};
		\node [style=none] (103) at (10.75, -1.25) {};
		\node [style=none] (104) at (10.75, -3.25) {};
		\node [style=none] (107) at (14.5, -1.25) {};
		\node [style=none] (108) at (14.5, -3.25) {};
		\node [style=Z dot] (109) at (12, -1.25) {};
		\node [style=Z dot] (110) at (12, -3.25) {};
		\node [style=X dot] (111) at (13.5, -1.25) {};
		\node [style=X dot] (112) at (13.5, -3.25) {};
		\node [style=none] (113) at (9.75, -2.25) {=};
		\node [style=none] (115) at (0.25, -1.3) {$\vdots$};
		\node [style=none] (116) at (0.25, -2.8) {$\vdots$};
		\node [style=none] (125) at (-5, 4.25) {$\vdots$};
		\node [style=none] (126) at (-0.75, 4.25) {$\vdots$};
		\node [style=none] (127) at (-0.75, 2.5) {$\vdots$};
		\node [style=none] (129) at (1.5, 4.25) {$\vdots$};
		\node [style=none] (130) at (1.5, 2.5) {$\vdots$};
		\node [style=none] (131) at (4.5, 2.5) {$\vdots$};
		\node [style=none] (132) at (4.5, 4.25) {$\vdots$};
		\node [style=none] (133) at (-5, 2.5) {$\vdots$};
	\end{pgfonlayer}
	\begin{pgfonlayer}{edgelayer}
		\draw [bend right=15] (0) to (2.center);
		\draw [bend left=15] (0) to (3.center);
		\draw [bend right=15] (1) to (4.center);
		\draw [bend left=15] (1) to (5.center);
		\draw [in=105, out=-75] (0) to (1);
		\draw [bend left=15] (1) to (8.center);
		\draw [bend right=15] (1) to (9.center);
		\draw [bend right=15] (0) to (7.center);
		\draw [bend left=15] (0) to (6.center);
		\draw [bend right=15] (15) to (16.center);
		\draw [bend left=15] (15) to (17.center);
		\draw [bend left=15] (15) to (18.center);
		\draw [bend right=15] (15) to (19.center);
		\draw [bend right=15] (23) to (30);
		\draw [bend right=15, looseness=0.75] (30) to (24.center);
		\draw (23) to (31);
		\draw (31) to (25.center);
		\draw [bend left=15] (23) to (32);
		\draw [bend left=15] (32) to (26.center);
		\draw [bend right] (23) to (33);
		\draw [bend right=15] (33) to (29.center);
		\draw (23) to (34);
		\draw (34) to (28.center);
		\draw [bend left=15] (23) to (35);
		\draw [bend left=15, looseness=1.25] (35) to (27.center);
		\draw [bend right] (38) to (40.center);
		\draw (39.center) to (38);
		\draw [bend left] (38) to (41.center);
		\draw [bend right] (38) to (44.center);
		\draw (43.center) to (38);
		\draw [bend left] (38) to (42.center);
		\draw (45.center) to (51);
		\draw (51) to (46.center);
		\draw (47.center) to (48.center);
		\draw (54.center) to (55.center);
		\draw (52.center) to (58);
		\draw (58) to (59);
		\draw (59) to (53.center);
		\draw (60.center) to (61);
		\draw (61) to (62);
		\draw [bend left] (62) to (63.center);
		\draw (62) to (64.center);
		\draw [bend right] (62) to (65.center);
		\draw (68.center) to (70);
		\draw [in=-165, out=75] (70) to (77);
		\draw (70) to (78);
		\draw [bend right] (70) to (79);
		\draw (79) to (73.center);
		\draw (78) to (72.center);
		\draw (77) to (71.center);
		\draw (82) to (83);
		\draw [bend left] (83) to (84.center);
		\draw (83) to (85.center);
		\draw [bend right] (83) to (86.center);
		\draw (91) to (94.center);
		\draw (92) to (95.center);
		\draw (93) to (96.center);
		\draw (99) to (100);
		\draw [bend right] (99) to (97.center);
		\draw [bend left] (99) to (98.center);
		\draw [bend left] (100) to (101.center);
		\draw [bend right] (100) to (102.center);
		\draw (103.center) to (109);
		\draw (111) to (107.center);
		\draw (112) to (108.center);
		\draw (110) to (104.center);
		\draw (109) to (112);
		\draw (110) to (111);
		\draw (109) to (111);
		\draw (110) to (112);
	\end{pgfonlayer}
\end{tikzpicture}

%% file: tikzpaper/XYmeasurement.tikz
\begin{tikzpicture}[baseline=-0.05]
	\begin{pgfonlayer}{nodelayer}
		\node [style=none] (0) at (-0.25, 0) {};
		\node [style=Z phase dot] (1) at (0.25, 0) {$\alpha(i)$};
	\end{pgfonlayer}
	\begin{pgfonlayer}{edgelayer}
		\draw (0.center) to (1);
	\end{pgfonlayer}
\end{tikzpicture}

%% file: tikzpaper/XZmeasurement.tikz
\begin{tikzpicture}[baseline=-0.05, scale=.8]
	\begin{pgfonlayer}{nodelayer}
		\node [style=none] (0) at (-0.25, 0) {};
		\node [style=Z phase dot] (1) at (0.25, 0) {$\frac{\pi}{2}$};
		\node [style=X phase dot] (2) at (1, 0) {$\alpha(i)$};
	\end{pgfonlayer}
	\begin{pgfonlayer}{edgelayer}
		\draw (0.center) to (1);
		\draw (1) to (2);
	\end{pgfonlayer}
\end{tikzpicture}

%% file: tikzpaper/YZmeasurement.tikz
\begin{tikzpicture}[baseline=-0.05]
	\begin{pgfonlayer}{nodelayer}
		\node [style=none] (0) at (-0.25, 0) {};
		\node [style=X phase dot] (1) at (0.25, 0) {$\alpha(i)$};
	\end{pgfonlayer}
	\begin{pgfonlayer}{edgelayer}
		\draw (0.center) to (1);
	\end{pgfonlayer}
\end{tikzpicture}

%% file: tikzpaper/Xmeasurement.tikz
\begin{tikzpicture}[baseline=-0.05]
	\begin{pgfonlayer}{nodelayer}
		\node [style=none] (0) at (-0.25, 0) {};
		\node [style=Z dot] (1) at (0.25, 0) {};
	\end{pgfonlayer}
	\begin{pgfonlayer}{edgelayer}
		\draw (0.center) to (1);
	\end{pgfonlayer}
\end{tikzpicture}

%% file: tikzpaper/Ymeasurement.tikz
\begin{tikzpicture}[baseline=-0.05]
	\begin{pgfonlayer}{nodelayer}
		\node [style=none] (0) at (-0.25, 0) {};
		\node [style=Z phase dot] (1) at (0.25, 0) {$\frac{\pi}{2}$};
	\end{pgfonlayer}
	\begin{pgfonlayer}{edgelayer}
		\draw (0.center) to (1);
	\end{pgfonlayer}
\end{tikzpicture}

%% file: tikzpaper/Zmeasurement.tikz
\begin{tikzpicture}[baseline=-0.05]
	\begin{pgfonlayer}{nodelayer}
		\node [style=none] (0) at (-0.25, 0) {};
		\node [style=X dot] (1) at (0.25, 0) {};
	\end{pgfonlayer}
	\begin{pgfonlayer}{edgelayer}
		\draw (0.center) to (1);
	\end{pgfonlayer}
\end{tikzpicture}

%% file: tikzpaper/Xmeasure2.tikz
\begin{tikzpicture}[baseline=-0.05]
	\begin{pgfonlayer}{nodelayer}
		\node [style=none] (0) at (-0.25, 0) {};
		\node [style=Z phase dot] (1) at (0.25, 0) {$\pi$};
	\end{pgfonlayer}
	\begin{pgfonlayer}{edgelayer}
		\draw (1) to (0.center);
	\end{pgfonlayer}
\end{tikzpicture}

%% file: tikzpaper/Y-measure2.tikz
\begin{tikzpicture}[baseline=-0.05]
	\begin{pgfonlayer}{nodelayer}
		\node [style=none] (0) at (-0.25, 0) {};
		\node [style=Z phase dot] (1) at (0.25, 0) {$-\frac{\pi}{2}$};
	\end{pgfonlayer}
	\begin{pgfonlayer}{edgelayer}
		\draw (1) to (0.center);
	\end{pgfonlayer}
\end{tikzpicture}

%% file: tikzpaper/Zmeasure2.tikz
\begin{tikzpicture}[baseline=-0.05]
	\begin{pgfonlayer}{nodelayer}
		\node [style=none] (0) at (-0.25, 0) {};
		\node [style=X phase dot] (1) at (0.25, 0) {$\pi$};
	\end{pgfonlayer}
	\begin{pgfonlayer}{edgelayer}
		\draw (1) to (0.center);
	\end{pgfonlayer}
\end{tikzpicture}

%% file: tikzpaper/Z-deletenew.tikz
\begin{tikzpicture}[scale=0.5]
	\begin{pgfonlayer}{nodelayer}
		\node [style=Z dot] (0) at (-3.25, 0) {};
		\node [style=Z dot] (1) at (-4.25, 1.25) {};
		\node [style=Z dot] (2) at (-4.25, -1.25) {};
		\node [style=Z dot] (3) at (-2.25, -1.25) {};
		\node [style=Z dot] (4) at (-2.25, 1.25) {};
		\node [style=none] (5) at (-1, 2) {};
		\node [style=none] (6) at (-1, 0.5) {};
		\node [style=none] (7) at (-1, -0.5) {};
		\node [style=none] (8) at (-1, -2) {};
		\node [style=none] (9) at (-5.75, 2) {};
		\node [style=none] (10) at (-5.75, 0.5) {};
		\node [style=none] (11) at (-5.75, -0.5) {};
		\node [style=none] (12) at (-5.75, -2) {};
		\node [style=none] (13) at (-5.75, 1.25) {$\vdots$};
		\node [style=none] (15) at (-1, -1.25) {$\vdots$};
		\node [style=none] (16) at (-1, 1.25) {$\vdots$};
		\node [style=none] (17) at (-5.75, -1.25) {$\vdots$};
		\node [style=X phase dot] (18) at (-3.25, 1.5) {$a\pi$};
		\node [style=none] (19) at (0, 0) {=};
		\node [style=none] (37) at (-4.25, 0.25) {$\vdots$};
		\node [style=none] (38) at (-2.25, 0.25) {$\vdots$};
		\node [style=Z phase dot] (42) at (2.25, 1.25) {$a\pi$};
		\node [style=Z phase dot] (43) at (2.25, -1.25) {$a\pi$};
		\node [style=Z phase dot] (44) at (4.25, -1.25) {$a\pi$};
		\node [style=Z phase dot] (45) at (4.25, 1.25) {$a\pi$};
		\node [style=none] (46) at (5.5, 2) {};
		\node [style=none] (47) at (5.5, 0.5) {};
		\node [style=none] (48) at (5.5, -0.5) {};
		\node [style=none] (49) at (5.5, -2) {};
		\node [style=none] (50) at (0.75, 2) {};
		\node [style=none] (51) at (0.75, 0.5) {};
		\node [style=none] (52) at (0.75, -0.5) {};
		\node [style=none] (53) at (0.75, -2) {};
		\node [style=none] (54) at (0.75, 1.25) {$\vdots$};
		\node [style=none] (55) at (5.5, -1.25) {$\vdots$};
		\node [style=none] (56) at (5.5, 1.25) {$\vdots$};
		\node [style=none] (57) at (0.75, -1.25) {$\vdots$};
		\node [style=none] (59) at (2.25, 0.25) {$\vdots$};
		\node [style=none] (60) at (4.25, 0.25) {$\vdots$};
	\end{pgfonlayer}
	\begin{pgfonlayer}{edgelayer}
		\draw (1) to (9.center);
		\draw (1) to (10.center);
		\draw (2) to (11.center);
		\draw (2) to (12.center);
		\draw (3) to (8.center);
		\draw (3) to (7.center);
		\draw (4) to (6.center);
		\draw (4) to (5.center);
		\draw (18) to (0);
		\draw [style=Hadamard edge] (1) to (0);
		\draw [style=Hadamard edge] (0) to (4);
		\draw [style=Hadamard edge] (0) to (3);
		\draw [style=Hadamard edge] (0) to (2);
		\draw (42) to (50.center);
		\draw (42) to (51.center);
		\draw (43) to (52.center);
		\draw (43) to (53.center);
		\draw (44) to (49.center);
		\draw (44) to (48.center);
		\draw (45) to (47.center);
		\draw (45) to (46.center);
	\end{pgfonlayer}
\end{tikzpicture}

%% file: tikzpaper/LC1.tikz
\begin{tikzpicture}[xscale=0.4,yscale=0.3]
	\begin{pgfonlayer}{nodelayer}
		\node [style=Z dot] (0) at (2, 3) {};
		\node [style=Z dot] (1) at (-1, 1) {};
		\node [style=Z dot] (2) at (1, -1) {};
		\node [style=Z dot] (3) at (-1, -3) {};
		\node [style=none] (4) at (4, 3) {};
		\node [style=none] (5) at (4, 1) {};
		\node [style=none] (6) at (4, -1) {};
		\node [style=none] (7) at (4, -3) {};
		\node [style=none] (8) at (1, -1.9) {$\vdots$};
		\node [style=none] (9) at (1, 3) {u};
		\node [style=none] (10) at (5, 0) {=};
	\end{pgfonlayer}
	\begin{pgfonlayer}{edgelayer}
		\draw [style=Hadamard edge] (0) to (3);
		\draw [style=Hadamard edge] (1) to (3);
		\draw [style=Hadamard edge] (1) to (0);
		\draw [style=Hadamard edge] (0) to (2);
		\draw (0) to (4.center);
		\draw (1) to (5.center);
		\draw (2) to (6.center);
		\draw (3) to (7.center);
	\end{pgfonlayer}
\end{tikzpicture}

%% file: tikzpaper/LC2.tikz
\begin{tikzpicture}[xscale=0.4,yscale=0.3]
	\begin{pgfonlayer}{nodelayer}
		\node [style=Z dot] (0) at (0, 3) {};
		\node [style=Z dot] (1) at (-3, 1) {};
		\node [style=Z dot] (2) at (-1, -1) {};
		\node [style=Z dot] (3) at (-3, -3) {};
		\node [style=none] (4) at (4, 3) {};
		\node [style=none] (5) at (4, 1) {};
		\node [style=none] (6) at (4, -1) {};
		\node [style=none] (7) at (4, -3) {};
		\node [style=none] (8) at (-1, -1.9) {$\vdots$};
		\node [style=X phase dot] (9) at (2, 3) {$-\frac{\pi}{2}$};
		\node [style=Z phase dot] (10) at (2, 1) {$\frac{\pi}{2}$};
		\node [style=Z phase dot] (11) at (2, -1) {$\frac{\pi}{2}$};
		\node [style=Z phase dot] (12) at (2, -3) {$\frac{\pi}{2}$};
		\node [style=none] (15) at (-1, 3) {u};
	\end{pgfonlayer}
	\begin{pgfonlayer}{edgelayer}
		\draw [style=Hadamard edge] (0) to (3);
		\draw [style=Hadamard edge] (1) to (0);
		\draw [style=Hadamard edge] (0) to (2);
		\draw (0) to (9);
		\draw (9) to (4.center);
		\draw (1) to (10);
		\draw (10) to (5.center);
		\draw (2) to (11);
		\draw (11) to (6.center);
		\draw (3) to (12);
		\draw (12) to (7.center);
		\draw [style=Hadamard edge] (1) to (2);
		\draw [style=Hadamard edge] (2) to (3);
	\end{pgfonlayer}
\end{tikzpicture}

%% file: tikzpaper/pivot1.tikz
\begin{tikzpicture}[scale=0.4]
	\begin{pgfonlayer}{nodelayer}
		\node [style=Z dot] (0) at (2, 3) {};
		\node [style=Z dot] (1) at (2, -3) {};
		\node [style=Z dot] (2) at (0, -1) {};
		\node [style=Z dot] (3) at (0, 1) {};
		\node [style=Z dot] (4) at (-2.5, 3.25) {};
		\node [style=Z dot] (5) at (-4, 2) {};
		\node [style=Z dot] (6) at (-4, -2) {};
		\node [style=Z dot] (7) at (-2.5, -3.25) {};
		\node [style=none] (8) at (-4.5, -2.5) {};
		\node [style=none] (9) at (-3, -3.75) {};
		\node [style=none] (10) at (-4.5, 2.5) {};
		\node [style=none] (11) at (-3, 3.75) {};
		\node [style=none] (12) at (5, 1) {};
		\node [style=none] (13) at (5, -1) {};
		\node [style=none] (14) at (5, 3) {};
		\node [style=none] (15) at (5, -3.25) {};
		\node [style=none] (16) at (6, 0) {=};
		\node [style=none] (17) at (2.5, 3.5) {$v$};
		\node [style=none] (18) at (2.5, -3.5) {$u$};
	\end{pgfonlayer}
	\begin{pgfonlayer}{edgelayer}
		\draw [style=Hadamard edge] (5) to (7);
		\draw [style=Hadamard edge] (6) to (4);
		\draw [style=Hadamard edge] (4) to (0);
		\draw [style=Hadamard edge] (5) to (0);
		\draw [style=Hadamard edge] (7) to (1);
		\draw [style=Hadamard edge] (6) to (1);
		\draw [style=Hadamard edge] (2) to (1);
		\draw [style=Hadamard edge] (2) to (0);
		\draw [style=Hadamard edge] (3) to (1);
		\draw [style=Hadamard edge] (0) to (1);
		\draw [style=Hadamard edge] (3) to (0);
		\draw [style=Hadamard edge] (2) to (4);
		\draw [style=Hadamard edge] (3) to (5);
		\draw [style=Hadamard edge] (3) to (6);
		\draw [style=Hadamard edge] (2) to (7);
		\draw (10.center) to (5);
		\draw (11.center) to (4);
		\draw (6) to (8.center);
		\draw (7) to (9.center);
		\draw (2) to (13.center);
		\draw (3) to (12.center);
		\draw [bend right=13] (0) to (15.center);
		\draw [bend left=13] (1) to (14.center);
	\end{pgfonlayer}
\end{tikzpicture}

%% file: tikzpaper/pivot2.tikz
\begin{tikzpicture}[scale=0.4]
	\begin{pgfonlayer}{nodelayer}
		\node [style=Z dot] (0) at (2, 3) {};
		\node [style=Z dot] (1) at (2, -3) {};
		\node [style=Z dot] (2) at (0, -1) {};
		\node [style=Z dot] (3) at (0, 1) {};
		\node [style=Z dot] (4) at (-2.5, 3.25) {};
		\node [style=Z dot] (5) at (-4, 2) {};
		\node [style=Z dot] (6) at (-4, -2) {};
		\node [style=Z dot] (7) at (-2.5, -3.25) {};
		\node [style=none] (8) at (-4.5, -2.5) {};
		\node [style=none] (9) at (-3, -3.75) {};
		\node [style=none] (10) at (-4.5, 2.5) {};
		\node [style=none] (11) at (-3, 3.75) {};
		\node [style=none] (12) at (5, 1) {};
		\node [style=none] (13) at (5, -1) {};
		\node [style=none] (14) at (5, 3) {};
		\node [style=H gate] (16) at (3.5, 3) {};
		\node [style=H gate] (17) at (3.5, -3) {};
		\node [style=none] (18) at (2.5, 3.5) {$u$};
		\node [style=none] (19) at (2.5, -3.5) {$v$};
		\node [style=Z phase dot] (20) at (3, 1) {$\pi$};
		\node [style=Z phase dot] (21) at (3, -1) {$\pi$};
		\node [style=none] (22) at (5, -3) {};
	\end{pgfonlayer}
	\begin{pgfonlayer}{edgelayer}
		\draw [style=Hadamard edge] (4) to (0);
		\draw [style=Hadamard edge] (5) to (0);
		\draw [style=Hadamard edge] (7) to (1);
		\draw [style=Hadamard edge] (6) to (1);
		\draw [style=Hadamard edge] (2) to (1);
		\draw [style=Hadamard edge] (2) to (0);
		\draw [style=Hadamard edge] (3) to (1);
		\draw [style=Hadamard edge] (0) to (1);
		\draw [style=Hadamard edge] (3) to (0);
		\draw (10.center) to (5);
		\draw (11.center) to (4);
		\draw (6) to (8.center);
		\draw (7) to (9.center);
		\draw (1) to (17);
		\draw (0) to (16);
		\draw (16) to (14.center);
		\draw [style=Hadamard edge] (5) to (6);
		\draw [style=Hadamard edge] (4) to (7);
		\draw [style=Hadamard edge] (6) to (2);
		\draw [style=Hadamard edge] (7) to (3);
		\draw [style=Hadamard edge] (5) to (2);
		\draw [style=Hadamard edge] (4) to (3);
		\draw (3) to (20);
		\draw (20) to (12.center);
		\draw (2) to (21);
		\draw (21) to (13.center);
		\draw (17) to (22.center);
	\end{pgfonlayer}
\end{tikzpicture}

%% file: tikzpaper/rGS-LC_new.tikz
\begin{tikzpicture}[scale=0.4]
	\begin{pgfonlayer}{nodelayer}
		\node [style=Z dot] (0) at (-3, 2.25) {};
		\node [style=Z dot] (1) at (-1, 0.75) {};
		\node [style=Z dot] (2) at (-1, -0.75) {};
		\node [style=Z dot] (3) at (-3, -2.25) {};
		\node [style=Z phase dot] (4) at (0.5, 2.25) {$\frac{\pi}{2}$};
		\node [style=Z phase dot] (5) at (0.5, 0.75) {$\frac{\pi}{2}$};
		\node [style=Z phase dot] (6) at (0.5, -0.75) {$-\frac{\pi}{2}$};
		\node [style=Z phase dot] (7) at (0.5, -2.25) {$\frac{\pi}{2}$};
		\node [style=X phase dot] (10) at (2.5, -0.75) {$\frac{\pi}{2}$};
		\node [style=X phase dot] (11) at (2.5, -2.25) {$\frac{\pi}{2}$};
		\node [style=none] (12) at (4, 2.25) {};
		\node [style=none] (13) at (4, 0.75) {};
		\node [style=none] (14) at (4, -0.75) {};
		\node [style=none] (15) at (4, -2.25) {};
	\end{pgfonlayer}
	\begin{pgfonlayer}{edgelayer}
		\draw [style=Hadamard edge] (0) to (3);
		\draw [style=Hadamard edge] (0) to (2);
		\draw [style=Hadamard edge] (3) to (1);
		\draw (0) to (4);
		\draw (1) to (5);
		\draw (2) to (6);
		\draw (6) to (10);
		\draw (3) to (7);
		\draw (7) to (11);
		\draw (4) to (12.center);
		\draw (5) to (13.center);
		\draw (10) to (14.center);
		\draw (11) to (15.center);
	\end{pgfonlayer}
\end{tikzpicture}

%% file: tikzpaper/nf_examplenew.tikz
\begin{tikzpicture}[scale=0.5]
	\begin{pgfonlayer}{nodelayer}
		\node [style=Z phase dot] (0) at (-3, 1.5) {$\frac{\pi}{2}$};
		\node [style=Z phase dot] (1) at (-1, 0.5) {};
		\node [style=X phase dot] (2) at (-1, -0.5) {$\pi$};
		\node [style=X dot] (3) at (-3, -1.5) {};
		\node [style=none] (4) at (2, 1.5) {};
		\node [style=none] (5) at (2, 0.5) {};
		\node [style=none] (6) at (2, -0.5) {};
		\node [style=none] (7) at (2, -1.5) {};
	\end{pgfonlayer}
	\begin{pgfonlayer}{edgelayer}
		\draw [style=Hadamard edge] (0) to (1);
		\draw (0) to (2);
		\draw (0) to (3);
		\draw (3) to (1);
		\draw (0) to (4.center);
		\draw (1) to (5.center);
		\draw (2) to (6.center);
		\draw (3) to (7.center);
	\end{pgfonlayer}
\end{tikzpicture}

%% file: tikzpaper/GSnf_examplenew.tikz
\begin{tikzpicture}[scale=0.5]
	\begin{pgfonlayer}{nodelayer}
		\node [style=Z dot] (0) at (-3, 1.5) {};
		\node [style=Z dot] (1) at (-1, 0.5) {};
		\node [style=none] (4) at (2, 1.5) {};
		\node [style=none] (5) at (2, 0.5) {};
		\node [style=none] (6) at (2, -0.5) {};
		\node [style=none] (7) at (2, -1.5) {};
		\node [style=Z dot] (8) at (-1, -0.5) {};
		\node [style=Z dot] (9) at (-3, -1.5) {};
		\node [style=H gate] (10) at (1.25, -0.5) {};
		\node [style=H gate] (11) at (0, -1.5) {};
		\node [style=Z phase dot] (12) at (0, 1.5) {$\frac{\pi}{2}$};
		\node [style=Z phase dot] (13) at (0, -0.5) {$\pi$};
	\end{pgfonlayer}
	\begin{pgfonlayer}{edgelayer}
		\draw [style=Hadamard edge] (0) to (1);
		\draw (1) to (5.center);
		\draw [style=Hadamard edge] (0) to (8);
		\draw [style=Hadamard edge] (9) to (0);
		\draw [style=Hadamard edge] (9) to (1);
		\draw (9) to (11);
		\draw (11) to (7.center);
		\draw (10) to (6.center);
		\draw (0) to (12);
		\draw (12) to (4.center);
		\draw (8) to (13);
		\draw (13) to (10);
	\end{pgfonlayer}
\end{tikzpicture}

%% file: tikzpaper/MBQCdiagrams2.tikz
\begin{tikzpicture}[scale=0.5]
	\begin{pgfonlayer}{nodelayer}
		\node [style=Z dot] (37) at (-8, 1) {};
		\node [style=Z dot] (38) at (-8, -1) {};
		\node [style=Z dot] (39) at (-6, 1) {};
		\node [style=Z dot] (40) at (-6, -1) {};
		\node [style=Z dot] (41) at (-9, 0) {};
		\node [style=Z dot] (42) at (-5, 0) {};
		\node [style=Z phase dot] (43) at (-5, 1) {$\frac{\pi}{2}$};
		\node [style=Z phase dot] (44) at (-5, -1) {$-\frac{\pi}{2}$};
		\node [style=Z phase dot] (45) at (-8, 2) {$\frac{\pi}{2}$};
		\node [style=Z phase dot] (46) at (-8, -2.25) {$\pi$};
		\node [style=Z phase dot] (47) at (-10, 0) {$\frac{\pi}{2}$};
		\node [style=Z phase dot] (48) at (-4, 0) {$\frac{\pi}{2}$};
		\node [style=X phase dot] (49) at (-4, 1) {$\frac{\pi}{2}$};
		\node [style=X phase dot] (50) at (-4, -1) {$\frac{\pi}{2}$};
		\node [style=X phase dot] (51) at (-9, 1) {$-\frac{\pi}{2}$};
		\node [style=Z phase dot] (52) at (-10, 1) {$\frac{\pi}{2}$};
		\node [style=none] (53) at (-11, 1) {};
		\node [style=none] (54) at (-11, -1) {};
		\node [style=none] (55) at (-3, 1) {};
		\node [style=none] (56) at (-3, -1) {};
		\node [style=Z dot] (57) at (4.5, 1) {};
		\node [style=Z dot] (58) at (4.5, -1) {};
		\node [style=Z dot] (59) at (6, 1.75) {};
		\node [style=Z dot] (60) at (7.5, 1) {};
		\node [style=Z dot] (61) at (7.5, -1) {};
		\node [style=none] (62) at (10.25, 1) {};
		\node [style=none] (63) at (10.25, -1) {};
		\node [style=none] (64) at (3, 1) {};
		\node [style=none] (65) at (3, -1) {};
		\node [style=hadamard] (66) at (3.75, -1) {};
		\node [style=Z dot] (67) at (6, -2) {};
		\node [style=Z phase dot] (68) at (7.5, -2) {$\frac{\pi}{2}$};
		\node [style=Z phase dot] (69) at (4.5, -1.75) {$\frac{\pi}{2}$};
		\node [style=Z phase dot] (70) at (8.5, 1) {$-\frac{\pi}{2}$};
		\node [style=Z phase dot] (71) at (8.5, -1) {$\frac{\pi}{2}$};
		\node [style=Z phase dot] (72) at (7.5, 1.75) {$\frac{\pi}{2}$};
		\node [style=hadamard] (73) at (9.5, 1) {};
		\node [style=Z dot] (74) at (4.5, 1.75) {};
	\end{pgfonlayer}
	\begin{pgfonlayer}{edgelayer}
		\draw [style=hadamard edge] (41) to (38);
		\draw [style=hadamard edge] (41) to (37);
		\draw [style=hadamard edge] (37) to (39);
		\draw [style=hadamard edge] (37) to (40);
		\draw [style=hadamard edge] (37) to (38);
		\draw [style=hadamard edge] (38) to (40);
		\draw [style=hadamard edge] (40) to (42);
		\draw [style=hadamard edge] (42) to (39);
		\draw [in=90, out=-90] (45) to (37);
		\draw (46) to (38);
		\draw (41) to (47);
		\draw (39) to (43);
		\draw (40) to (44);
		\draw (42) to (48);
		\draw (53.center) to (52);
		\draw (52) to (51);
		\draw (51) to (37);
		\draw (54.center) to (38);
		\draw (43) to (49);
		\draw (49) to (55.center);
		\draw (44) to (50);
		\draw (50) to (56.center);
		\draw [style=hadamard edge] (57) to (59);
		\draw [style=hadamard edge] (59) to (60);
		\draw [style=hadamard edge] (57) to (61);
		\draw (64.center) to (57);
		\draw (65.center) to (66);
		\draw (66) to (58);
		\draw (58) to (69);
		\draw (67) to (68);
		\draw [style=hadamard edge] (58) to (67);
		\draw [style=hadamard edge] (67) to (61);
		\draw (60) to (70);
		\draw (61) to (71);
		\draw (72) to (59);
		\draw [style=hadamard edge] (57) to (58);
		\draw (60) to (70);
		\draw (70) to (73);
		\draw (73) to (62.center);
		\draw (71) to (63.center);
		\draw (74) to (57);
	\end{pgfonlayer}
\end{tikzpicture}

%% file: tikzpaper/rewriteprocedure4.tikz
\begin{tikzpicture}[scale=0.5]
	\begin{pgfonlayer}{nodelayer}
		\node [style=Z dot] (0) at (-4, 1) {};
		\node [style=Z dot] (1) at (-4, -1) {};
		\node [style=Z dot] (2) at (-2.75, 1) {};
		\node [style=Z dot] (3) at (-2.75, -1) {};
		\node [style=Z dot] (4) at (-5, 0) {};
		\node [style=Z dot] (5) at (-2, 0) {};
		\node [style=Z phase dot] (6) at (-1.75, 1) {$\frac{\pi}{2}$};
		\node [style=Z phase dot] (7) at (-1.75, -1) {$-\frac{\pi}{2}$};
		\node [style=Z phase dot] (8) at (-4, 2) {$\frac{\pi}{2}$};
		\node [style=Z phase dot] (9) at (-4, -2) {$\pi$};
		\node [style=Z phase dot] (10) at (-6, 0) {$\frac{\pi}{2}$};
		\node [style=Z phase dot] (11) at (-1, 0) {$\frac{\pi}{2}$};
		\node [style=X phase dot] (12) at (-0.75, 1) {$\frac{\pi}{2}$};
		\node [style=X phase dot] (13) at (-0.75, -1) {$\frac{\pi}{2}$};
		\node [style=X phase dot] (14) at (-5, 1) {$-\frac{\pi}{2}$};
		\node [style=Z phase dot] (15) at (-6.25, 1) {$\frac{\pi}{2}$};
		\node [style=none] (16) at (-6.75, 1) {};
		\node [style=none] (17) at (-6.75, -1) {};
		\node [style=none] (18) at (0, 1) {};
		\node [style=none] (19) at (0, -1) {};
		\node [style=Z dot] (20) at (3.25, 1) {};
		\node [style=Z dot] (21) at (3.25, -1) {};
		\node [style=Z dot] (22) at (4.5, 1) {};
		\node [style=Z dot] (23) at (4.5, -1) {};
		\node [style=Z dot] (24) at (2.5, 0) {};
		\node [style=Z dot] (25) at (5.5, 0) {};
		\node [style=Z phase dot] (27) at (5.5, -1) {$\pi$};
		\node [style=Z phase dot] (29) at (3.25, -2) {$\frac{\pi}{2}$};
		\node [style=X phase dot] (32) at (6, 1) {$\frac{\pi}{2}$};
		\node [style=X phase dot] (33) at (6.5, -1) {$\frac{\pi}{2}$};
		\node [style=X phase dot] (34) at (2.25, 1) {$-\frac{\pi}{2}$};
		\node [style=Z phase dot] (35) at (1, 1) {$\frac{\pi}{2}$};
		\node [style=none] (36) at (0.5, 1) {};
		\node [style=none] (37) at (0.5, -1) {};
		\node [style=none] (38) at (7.25, 1) {};
		\node [style=none] (39) at (7.25, -1) {};
		\node [style=X dot] (40) at (1.5, 0) {};
		\node [style=X dot] (41) at (6.5, 0) {};
		\node [style=Z dot] (42) at (10.5, 1) {};
		\node [style=Z dot] (43) at (10.5, -1) {};
		\node [style=Z dot] (44) at (11.75, 1) {};
		\node [style=Z dot] (45) at (11.75, -1) {};
		\node [style=Z phase dot] (48) at (12.5, -1) {$\pi$};
		\node [style=Z phase dot] (50) at (10.5, -2) {$\frac{\pi}{2}$};
		\node [style=X phase dot] (51) at (13.25, 1) {$\frac{\pi}{2}$};
		\node [style=X phase dot] (52) at (13.5, -1) {$\frac{\pi}{2}$};
		\node [style=X phase dot] (53) at (9.5, 1) {$-\frac{\pi}{2}$};
		\node [style=Z phase dot] (54) at (8.25, 1) {$\frac{\pi}{2}$};
		\node [style=none] (55) at (7.75, 1) {};
		\node [style=none] (56) at (7.75, -1) {};
		\node [style=none] (57) at (14.5, 1) {};
		\node [style=none] (58) at (14.5, -1) {};
		\node [style=Z dot] (59) at (15.75, 1) {};
		\node [style=Z dot] (60) at (15.75, -1) {};
		\node [style=Z dot] (61) at (17, 1) {};
		\node [style=Z dot] (62) at (17, -1) {};
		\node [style=Z phase dot] (63) at (18, -1) {$\pi$};
		\node [style=Z phase dot] (64) at (15.75, 2) {$\frac{\pi}{2}$};
		\node [style=Z phase dot] (65) at (15.75, -2.25) {$\frac{\pi}{2}$};
		\node [style=X phase dot] (66) at (19.25, 1) {$\frac{\pi}{2}$};
		\node [style=X phase dot] (67) at (19.25, -1) {$\frac{\pi}{2}$};
		\node [style=none] (70) at (15, 1) {};
		\node [style=none] (71) at (15, -1) {};
		\node [style=none] (72) at (20.25, 1) {};
		\node [style=none] (73) at (20.25, -1) {};
		\node [style=Z phase dot] (74) at (18, 1) {$-\frac{\pi}{2}$};
		\node [style=none] (75) at (14.75, 0) {$=$};
		\node [style=none] (76) at (7.5, 0) {$=$};
		\node [style=none] (77) at (0.25, 0) {$=$};
		\node [style=Z dot] (78) at (21.5, 1) {};
		\node [style=Z dot] (79) at (21.5, -1) {};
		\node [style=Z dot] (80) at (23, 1) {};
		\node [style=Z dot] (81) at (23, -1) {};
		\node [style=Z phase dot] (83) at (21.5, 2) {$\frac{\pi}{2}$};
		\node [style=Z dot] (84) at (21.5, -2.25) {};
		\node [style=none] (87) at (20.75, 1) {};
		\node [style=none] (88) at (20.75, -1) {};
		\node [style=none] (89) at (25.25, 1) {};
		\node [style=none] (90) at (25.25, -1) {};
		\node [style=Z phase dot] (91) at (23.75, 1) {$\pi$};
		\node [style=none] (92) at (20.5, 0) {$=$};
		\node [style=hadamard] (93) at (24.5, 1) {};
		\node [style=hadamard] (94) at (24, -1) {};
		\node [style=Z dot] (95) at (10.5, 2) {};
		\node [style=Z dot] (96) at (3.25, 2) {};
	\end{pgfonlayer}
	\begin{pgfonlayer}{edgelayer}
		\draw [style=hadamard edge] (4) to (1);
		\draw [style=hadamard edge] (4) to (0);
		\draw [style=hadamard edge] (0) to (2);
		\draw [style=hadamard edge] (0) to (3);
		\draw [style=hadamard edge] (0) to (1);
		\draw [style=hadamard edge] (1) to (3);
		\draw [style=hadamard edge] (3) to (5);
		\draw [style=hadamard edge] (5) to (2);
		\draw (8) to (0);
		\draw (9) to (1);
		\draw (4) to (10);
		\draw (2) to (6);
		\draw (3) to (7);
		\draw (5) to (11);
		\draw (16.center) to (15);
		\draw (15) to (14);
		\draw (14) to (0);
		\draw (17.center) to (1);
		\draw (6) to (12);
		\draw (12) to (18.center);
		\draw (7) to (13);
		\draw (13) to (19.center);
		\draw [style=hadamard edge] (24) to (21);
		\draw [style=hadamard edge] (24) to (20);
		\draw [style=hadamard edge] (20) to (22);
		\draw [style=hadamard edge] (20) to (23);
		\draw [style=hadamard edge] (21) to (23);
		\draw [style=hadamard edge] (23) to (25);
		\draw [style=hadamard edge] (25) to (22);
		\draw (29) to (21);
		\draw (23) to (27);
		\draw (36.center) to (35);
		\draw [in=180, out=0] (35) to (34);
		\draw (34) to (20);
		\draw (37.center) to (21);
		\draw (32) to (38.center);
		\draw (27) to (33);
		\draw (33) to (39.center);
		\draw (40) to (24);
		\draw (25) to (41);
		\draw [style=hadamard edge] (22) to (23);
		\draw (22) to (32);
		\draw [style=hadamard edge] (42) to (44);
		\draw [style=hadamard edge] (42) to (45);
		\draw [style=hadamard edge] (43) to (45);
		\draw (50) to (43);
		\draw (45) to (48);
		\draw (55.center) to (54);
		\draw (54) to (53);
		\draw (53) to (42);
		\draw (56.center) to (43);
		\draw (51) to (57.center);
		\draw (48) to (52);
		\draw (52) to (58.center);
		\draw [style=hadamard edge] (44) to (45);
		\draw (44) to (51);
		\draw [style=hadamard edge] (59) to (61);
		\draw [style=hadamard edge] (59) to (62);
		\draw [style=hadamard edge] (60) to (62);
		\draw (64) to (59);
		\draw (65) to (60);
		\draw (62) to (63);
		\draw (71.center) to (60);
		\draw (66) to (72.center);
		\draw (63) to (67);
		\draw (67) to (73.center);
		\draw (70.center) to (59);
		\draw [style=hadamard edge] (78) to (80);
		\draw [style=hadamard edge] (79) to (81);
		\draw (83) to (78);
		\draw (84) to (79);
		\draw (88.center) to (79);
		\draw (87.center) to (78);
		\draw [style=hadamard edge] (78) to (79);
		\draw [style=hadamard edge] (78) to (81);
		\draw (80) to (91);
		\draw (91) to (93);
		\draw (93) to (89.center);
		\draw (81) to (94);
		\draw (94) to (90.center);
		\draw (95) to (42);
		\draw (96) to (20);
		\draw (61) to (74);
		\draw (74) to (66);
	\end{pgfonlayer}
\end{tikzpicture}

%% file: tikzpaper/rewriteprocedure3.tikz
\begin{tikzpicture}[scale=0.5]
	\begin{pgfonlayer}{nodelayer}
		\node [style=Z dot] (18) at (-12.75, 1) {};
		\node [style=Z dot] (19) at (-12.75, -1) {};
		\node [style=Z dot] (20) at (-11.25, 1.75) {};
		\node [style=Z dot] (21) at (-9.75, 1) {};
		\node [style=Z dot] (22) at (-9.75, -1) {};
		\node [style=none] (24) at (-7.25, 1) {};
		\node [style=none] (25) at (-7.25, -1) {};
		\node [style=none] (26) at (-14.25, 1) {};
		\node [style=none] (27) at (-14.25, -1) {};
		\node [style=hadamard] (28) at (-13.5, -1) {};
		\node [style=Z dot] (30) at (-11.25, -2) {};
		\node [style=Z phase dot] (31) at (-9.75, -2) {$\frac{\pi}{2}$};
		\node [style=Z phase dot] (32) at (-12.75, -1.75) {$\frac{\pi}{2}$};
		\node [style=Z phase dot] (33) at (-8.75, 1) {$-\frac{\pi}{2}$};
		\node [style=Z phase dot] (34) at (-8.75, -1) {$\frac{\pi}{2}$};
		\node [style=Z phase dot] (85) at (-9.75, 1.75) {$\frac{\pi}{2}$};
		\node [style=none] (90) at (-6.75, 0) {$=$};
		\node [style=hadamard] (93) at (-8, 1) {};
		\node [style=Z dot] (94) at (-4.5, 1) {};
		\node [style=Z dot] (95) at (-4.5, -1) {};
		\node [style=Z dot] (96) at (-3, 1.75) {};
		\node [style=Z dot] (97) at (-1.5, 1) {};
		\node [style=Z dot] (98) at (-1.5, -1) {};
		\node [style=none] (99) at (1, 1) {};
		\node [style=none] (100) at (1, -1) {};
		\node [style=none] (101) at (-6, 1) {};
		\node [style=none] (102) at (-6, -1) {};
		\node [style=hadamard] (103) at (-5.25, -1) {};
		\node [style=Z phase dot] (104) at (-4.5, 1.75) {$-\frac{\pi}{2}$};
		\node [style=Z dot] (105) at (-3, -2) {};
		\node [style=Z phase dot] (106) at (-1.5, -2) {$\frac{\pi}{2}$};
		\node [style=Z phase dot] (107) at (-4.5, -1.75) {$\frac{\pi}{2}$};
		\node [style=Z phase dot] (108) at (-0.5, 1) {$\pi$};
		\node [style=Z phase dot] (109) at (-0.5, -1) {$\frac{\pi}{2}$};
		\node [style=none] (111) at (1.5, 0) {$=$};
		\node [style=hadamard] (112) at (0.25, 1) {};
		\node [style=X dot] (113) at (-1.5, 1.75) {};
		\node [style=Z dot] (114) at (3.5, 1) {};
		\node [style=Z dot] (115) at (3.5, -1) {};
		\node [style=Z dot] (116) at (5, 1.75) {};
		\node [style=Z dot] (117) at (6.5, 1) {};
		\node [style=Z dot] (118) at (6.5, -1) {};
		\node [style=none] (119) at (9, 1) {};
		\node [style=none] (120) at (9, -1) {};
		\node [style=none] (121) at (2, 1) {};
		\node [style=none] (122) at (2, -1) {};
		\node [style=hadamard] (123) at (2.75, -1) {};
		\node [style=Z phase dot] (124) at (3.5, 1.75) {$-\frac{\pi}{2}$};
		\node [style=Z dot] (125) at (5, -2) {};
		\node [style=Z phase dot] (128) at (7.5, 1) {$\pi$};
		\node [style=none] (130) at (9.5, 0) {$=$};
		\node [style=hadamard] (131) at (8.25, 1) {};
		\node [style=X dot] (132) at (6.25, 1.75) {};
		\node [style=X dot] (133) at (6.25, -2) {};
		\node [style=Z dot] (134) at (3.5, -1.75) {};
		\node [style=Z dot] (135) at (11.5, 1) {};
		\node [style=Z dot] (136) at (11.5, -1) {};
		\node [style=Z dot] (138) at (14.5, 1) {};
		\node [style=Z dot] (139) at (14.5, -1) {};
		\node [style=none] (140) at (17.25, 1) {};
		\node [style=none] (141) at (17.25, -1) {};
		\node [style=none] (142) at (10, 1) {};
		\node [style=none] (143) at (10, -1) {};
		\node [style=hadamard] (144) at (10.75, -1) {};
		\node [style=Z phase dot] (145) at (11.5, 1.75) {$-\frac{\pi}{2}$};
		\node [style=Z phase dot] (147) at (15.5, 1) {$\pi$};
		\node [style=hadamard] (149) at (16.5, 1) {};
		\node [style=Z dot] (152) at (11.5, -1.75) {};
		\node [style=Z dot] (153) at (-12.75, 1.75) {};
	\end{pgfonlayer}
	\begin{pgfonlayer}{edgelayer}
		\draw [style=hadamard edge] (18) to (20);
		\draw [style=hadamard edge] (20) to (21);
		\draw [style=hadamard edge] (18) to (22);
		\draw (26.center) to (18);
		\draw (27.center) to (28);
		\draw (28) to (19);
		\draw (19) to (32);
		\draw (30) to (31);
		\draw [style=hadamard edge] (19) to (30);
		\draw [style=hadamard edge] (30) to (22);
		\draw (21) to (33);
		\draw (22) to (34);
		\draw (85) to (20);
		\draw [style=hadamard edge] (18) to (19);
		\draw (21) to (33);
		\draw (33) to (93);
		\draw (93) to (24.center);
		\draw (34) to (25.center);
		\draw [style=hadamard edge] (94) to (96);
		\draw [style=hadamard edge] (96) to (97);
		\draw [style=hadamard edge] (94) to (98);
		\draw (101.center) to (94);
		\draw (102.center) to (103);
		\draw (103) to (95);
		\draw (104) to (94);
		\draw (95) to (107);
		\draw (105) to (106);
		\draw [style=hadamard edge] (95) to (105);
		\draw [style=hadamard edge] (105) to (98);
		\draw (97) to (108);
		\draw (98) to (109);
		\draw [style=hadamard edge] (94) to (95);
		\draw (97) to (108);
		\draw (108) to (112);
		\draw (112) to (99.center);
		\draw (109) to (100.center);
		\draw (96) to (113);
		\draw [style=hadamard edge] (94) to (97);
		\draw [style=hadamard edge] (114) to (116);
		\draw [style=hadamard edge] (116) to (117);
		\draw [style=hadamard edge] (114) to (118);
		\draw (121.center) to (114);
		\draw (122.center) to (123);
		\draw (123) to (115);
		\draw (124) to (114);
		\draw [style=hadamard edge] (115) to (125);
		\draw [style=hadamard edge] (125) to (118);
		\draw (117) to (128);
		\draw [style=hadamard edge] (114) to (115);
		\draw (117) to (128);
		\draw (128) to (131);
		\draw (131) to (119.center);
		\draw (116) to (132);
		\draw [style=hadamard edge] (114) to (117);
		\draw [style=hadamard edge] (125) to (133);
		\draw (115) to (134);
		\draw (118) to (120.center);
		\draw [style=hadamard edge] (135) to (139);
		\draw (142.center) to (135);
		\draw (143.center) to (144);
		\draw (144) to (136);
		\draw (145) to (135);
		\draw (138) to (147);
		\draw [style=hadamard edge] (135) to (136);
		\draw (138) to (147);
		\draw (147) to (149);
		\draw (149) to (140.center);
		\draw [style=hadamard edge] (135) to (138);
		\draw (136) to (152);
		\draw (139) to (141.center);
		\draw [style=hadamard edge] (115) to (118);
		\draw [style=hadamard edge] (136) to (139);
		\draw (153) to (18);
	\end{pgfonlayer}
\end{tikzpicture}

%% file: tikzpaper/canonical3.tikz
\begin{tikzpicture}[scale=0.5]
	\begin{pgfonlayer}{nodelayer}
		\node [style=Z dot] (0) at (-2, 1) {};
		\node [style=Z dot] (1) at (-2, -1) {};
		\node [style=Z dot] (2) at (1, 1) {};
		\node [style=Z dot] (3) at (1, -1) {};
		\node [style=none] (6) at (3.5, 1) {};
		\node [style=none] (7) at (3.5, -1) {};
		\node [style=none] (8) at (-3.5, 1) {};
		\node [style=none] (9) at (-3.5, -1) {};
		\node [style=Z phase dot] (11) at (-2, 1.75) {$\frac{\pi}{2}$};
		\node [style=hadamard] (13) at (2.25, -1) {};
		\node [style=Z dot] (14) at (-2, -1.75) {};
		\node [style=Z phase dot] (15) at (1.75, 1) {$\pi$};
		\node [style=hadamard] (16) at (2.75, 1) {};
	\end{pgfonlayer}
	\begin{pgfonlayer}{edgelayer}
		\draw [style=hadamard edge] (0) to (3);
		\draw (8.center) to (0);
		\draw (11) to (0);
		\draw (3) to (13);
		\draw (13) to (7.center);
		\draw (1) to (14);
		\draw (9.center) to (1);
		\draw [style=hadamard edge] (0) to (1);
		\draw (2) to (15);
		\draw (15) to (16);
		\draw (16) to (6.center);
		\draw [style=hadamard edge] (0) to (2);
		\draw [style=hadamard edge] (1) to (3);
	\end{pgfonlayer}
\end{tikzpicture}

%% file: tikzpaper/nf-example-numbered.tikz
\begin{tikzpicture}[scale=0.5]
	\begin{pgfonlayer}{nodelayer}
		\node [style=Z phase dot] (0) at (-3, 1.5) {$\frac{\pi}{2}$};
		\node [style=Z phase dot] (1) at (-1, 0.5) {};
		\node [style=X phase dot] (2) at (-1, -0.5) {$\pi$};
		\node [style=X dot] (3) at (-3, -1.5) {};
		\node [style=none] (4) at (0.75, 1.5) {};
		\node [style=none] (5) at (0.75, 0.5) {};
		\node [style=none] (6) at (0.75, -0.5) {};
		\node [style=none] (7) at (0.75, -1.5) {};
		\node [style=none] (8) at (1.25, 1.5) {$\scriptstyle 1$};
		\node [style=none] (9) at (1.25, 0.5) {$\scriptstyle 2$};
		\node [style=none] (10) at (1.25, -0.5) {$\scriptstyle 3$};
		\node [style=none] (11) at (1.25, -1.5) {$\scriptstyle 4$};
	\end{pgfonlayer}
	\begin{pgfonlayer}{edgelayer}
		\draw [style=Hadamard edge] (0) to (1);
		\draw (0) to (2);
		\draw (0) to (3);
		\draw (3) to (1);
		\draw (0) to (4.center);
		\draw (1) to (5.center);
		\draw (2) to (6.center);
		\draw (3) to (7.center);
	\end{pgfonlayer}
\end{tikzpicture}

%% file: tikzpaper/example38.tikz
\begin{tikzpicture}[scale=0.5]
	\begin{pgfonlayer}{nodelayer}
		\node [style=Z phase dot] (1) at (-1, 0.5) {$\pi$};
		\node [style=none] (4) at (1.5, 1.75) {};
		\node [style=none] (5) at (1.5, 0.5) {};
		\node [style=none] (6) at (1.5, -0.75) {};
		\node [style=none] (7) at (1.5, -2) {};
		\node [style=Z phase dot] (8) at (-1, -0.75) {$- \frac{\pi}{2}$};
		\node [style=X phase dot] (9) at (-2.75, -2) {$\pi$};
		\node [style=X phase dot] (10) at (-2.75, 1.75) {$\pi$};
		\node [style=none] (11) at (2.25, 1.75) {$\scriptstyle 1$};
		\node [style=none] (12) at (2.25, 0.5) {$\scriptstyle 2$};
		\node [style=none] (13) at (2.25, -0.75) {$\scriptstyle 3$};
		\node [style=none] (14) at (2.25, -2) {$\scriptstyle 4$};
	\end{pgfonlayer}
	\begin{pgfonlayer}{edgelayer}
		\draw (1) to (5.center);
		\draw (10) to (4.center);
		\draw (10) to (8);
		\draw (9) to (1);
		\draw (9) to (8);
		\draw (9) to (7.center);
		\draw (8) to (6.center);
		\draw [style=hadamard edge] (8) to (1);
	\end{pgfonlayer}
\end{tikzpicture}

%% file: main.bbl
\begin{thebibliography}{10}
\providecommand{\bibitemdeclare}[2]{}
\providecommand{\surnamestart}{}
\providecommand{\surnameend}{}
\providecommand{\urlprefix}{Available at }
\providecommand{\url}[1]{\texttt{#1}}
\providecommand{\href}[2]{\texttt{#2}}
\providecommand{\urlalt}[2]{\href{#1}{#2}}
\providecommand{\doi}[1]{doi:\urlalt{https://doi.org/#1}{#1}}
\providecommand{\eprint}[1]{arXiv:\urlalt{https://arxiv.org/abs/#1}{#1}}
\providecommand{\bibinfo}[2]{#2}

\bibitemdeclare{article}{amy_polynomial-time_2014}
\bibitem{amy_polynomial-time_2014}
\bibinfo{author}{Matthew \surnamestart Amy\surnameend}, \bibinfo{author}{Dmitri
  \surnamestart Maslov\surnameend} \& \bibinfo{author}{Michele \surnamestart
  Mosca\surnameend} (\bibinfo{year}{2014}):
  \emph{\bibinfo{title}{Polynomial-{Time} {T}-{Depth} {Optimization} of
  {Clifford}+{T} {Circuits} {Via} {Matroid} {Partitioning}}}.
\newblock {\slshape \bibinfo{journal}{IEEE Transactions on Computer-Aided
  Design of Integrated Circuits and Systems}}
  \bibinfo{volume}{33}(\bibinfo{number}{10}), pp. \bibinfo{pages}{1476--1489},
  \doi{10.1109/TCAD.2014.2341953}.

\bibitemdeclare{article}{Backens_2014}
\bibitem{Backens_2014}
\bibinfo{author}{Miriam \surnamestart Backens\surnameend}
  (\bibinfo{year}{2014}): \emph{\bibinfo{title}{The {ZX}-calculus is complete
  for stabilizer quantum mechanics}}.
\newblock {\slshape \bibinfo{journal}{New Journal of Physics}}
  \bibinfo{volume}{16}(\bibinfo{number}{9}), p. \bibinfo{pages}{093021},
  \doi{10.1088/1367-2630/16/9/093021}.

\bibitemdeclare{article}{Backens_2015}
\bibitem{Backens_2015}
\bibinfo{author}{Miriam \surnamestart Backens\surnameend}
  (\bibinfo{year}{2015}): \emph{\bibinfo{title}{Making the stabilizer
  {ZX}-calculus complete for scalars}}.
\newblock {\slshape \bibinfo{journal}{Electronic Proceedings in Theoretical
  Computer Science}} \bibinfo{volume}{195}, p. \bibinfo{pages}{17–32},
  \doi{10.4204/eptcs.195.2}.

\bibitemdeclare{article}{Backens_2021}
\bibitem{Backens_2021}
\bibinfo{author}{Miriam \surnamestart Backens\surnameend},
  \bibinfo{author}{Hector \surnamestart Miller-Bakewell\surnameend},
  \bibinfo{author}{Giovanni \surnamestart de~Felice\surnameend},
  \bibinfo{author}{Leo \surnamestart Lobski\surnameend} \&
  \bibinfo{author}{John \surnamestart van~de Wetering\surnameend}
  (\bibinfo{year}{2021}): \emph{\bibinfo{title}{There and back again: {A}
  circuit extraction tale}}.
\newblock {\slshape \bibinfo{journal}{Quantum}} \bibinfo{volume}{5}, p.
  \bibinfo{pages}{421}, \doi{10.22331/q-2021-03-25-421}.

\bibitemdeclare{misc}{deBeaudrap_2022}
\bibitem{deBeaudrap_2022}
\bibinfo{author}{Niel \surnamestart de~Beaudrap\surnameend},
  \bibinfo{author}{Aleks \surnamestart Kissinger\surnameend} \&
  \bibinfo{author}{John \surnamestart van~de Wetering\surnameend}
  (\bibinfo{year}{2022}): \emph{\bibinfo{title}{Circuit {Extraction} for
  {ZX}-diagrams can be $\#${P}-hard}}, \doi{10.48550/ARXIV.2202.09194}.
\newblock \urlprefix\url{https://arxiv.org/abs/2202.09194}.

\bibitemdeclare{article}{Bouchet_1988}
\bibitem{Bouchet_1988}
\bibinfo{author}{Andr\'{e} \surnamestart Bouchet\surnameend}
  (\bibinfo{year}{1987}): \emph{\bibinfo{title}{Graphic Presentations of
  Isotropic Systems}}.
\newblock {\slshape \bibinfo{journal}{Journal of Combinatorial Theory, Series
  B}} \bibinfo{volume}{45}(\bibinfo{number}{1}), p. \bibinfo{pages}{58–76},
  \doi{10.1016/0095-8956(88)90055-X}.

\bibitemdeclare{article}{broadbent_parallelizing_2009}
\bibitem{broadbent_parallelizing_2009}
\bibinfo{author}{Anne \surnamestart Broadbent\surnameend} \&
  \bibinfo{author}{Elham \surnamestart Kashefi\surnameend}
  (\bibinfo{year}{2009}): \emph{\bibinfo{title}{Parallelizing quantum
  circuits}}.
\newblock {\slshape \bibinfo{journal}{Theoretical Computer Science}}
  \bibinfo{volume}{410}(\bibinfo{number}{26}), pp. \bibinfo{pages}{2489--2510},
  \doi{10.1016/j.tcs.2008.12.046}.

\bibitemdeclare{article}{Browne_2007}
\bibitem{Browne_2007}
\bibinfo{author}{Daniel~E \surnamestart Browne\surnameend},
  \bibinfo{author}{Elham \surnamestart Kashefi\surnameend},
  \bibinfo{author}{Mehdi \surnamestart Mhalla\surnameend} \&
  \bibinfo{author}{Simon \surnamestart Perdrix\surnameend}
  (\bibinfo{year}{2007}): \emph{\bibinfo{title}{Generalized flow and
  determinism in measurement-based quantum computation}}.
\newblock {\slshape \bibinfo{journal}{New Journal of Physics}}
  \bibinfo{volume}{9}(\bibinfo{number}{8}), p. \bibinfo{pages}{250–250},
  \doi{10.1088/1367-2630/9/8/250}.

\bibitemdeclare{inproceedings}{cai_dichotomy_2018}
\bibitem{cai_dichotomy_2018}
\bibinfo{author}{Jin-Yi \surnamestart Cai\surnameend}, \bibinfo{author}{Pinyan
  \surnamestart Lu\surnameend} \& \bibinfo{author}{Mingji \surnamestart
  Xia\surnameend} (\bibinfo{year}{2018}): \emph{\bibinfo{title}{Dichotomy for
  {Real} {Holant}{\textasciicircum}c {Problems}}}.
\newblock In: {\slshape \bibinfo{booktitle}{Proceedings of the {Twenty}-{Ninth}
  {Annual} {ACM}-{SIAM} {Symposium} on {Discrete} {Algorithms}}},
  \bibinfo{publisher}{Society for Industrial and Applied Mathematics}, pp.
  \bibinfo{pages}{1802--1821}, \doi{10.1137/1.9781611975031.118}.

\bibitemdeclare{book}{Coecke_2017}
\bibitem{Coecke_2017}
\bibinfo{author}{Bob \surnamestart Coecke\surnameend} \& \bibinfo{author}{Aleks
  \surnamestart Kissinger\surnameend} (\bibinfo{year}{2017}):
  \emph{\bibinfo{title}{Picturing {Quantum} {Processes}: {A} {First} {Course}
  in {Quantum} {Theory} and {Diagrammatic} {Reasoning}}}.
\newblock \bibinfo{publisher}{Cambridge University Press},
  \doi{10.1017/9781316219317}.

\bibitemdeclare{article}{Danos_2006}
\bibitem{Danos_2006}
\bibinfo{author}{Vincent \surnamestart Danos\surnameend} \&
  \bibinfo{author}{Elham \surnamestart Kashefi\surnameend}
  (\bibinfo{year}{2006}): \emph{\bibinfo{title}{Determinism in the one-way
  model}}.
\newblock {\slshape \bibinfo{journal}{Phys. Rev. A}} \bibinfo{volume}{74}, p.
  \bibinfo{pages}{052310}, \doi{10.1103/PhysRevA.74.052310}.

\bibitemdeclare{article}{danos_parsimonious_2005}
\bibitem{danos_parsimonious_2005}
\bibinfo{author}{Vincent \surnamestart Danos\surnameend},
  \bibinfo{author}{Elham \surnamestart Kashefi\surnameend} \&
  \bibinfo{author}{Prakash \surnamestart Panangaden\surnameend}
  (\bibinfo{year}{2005}): \emph{\bibinfo{title}{Parsimonious and robust
  realizations of unitary maps in the one-way model}}.
\newblock {\slshape \bibinfo{journal}{Physical Review A}}
  \bibinfo{volume}{72}(\bibinfo{number}{6}), p. \bibinfo{pages}{064301},
  \doi{10.1103/PhysRevA.72.064301}.

\bibitemdeclare{article}{Dehaene_2003}
\bibitem{Dehaene_2003}
\bibinfo{author}{Jeroen \surnamestart Dehaene\surnameend} \&
  \bibinfo{author}{Bart \surnamestart De~Moor\surnameend}
  (\bibinfo{year}{2003}): \emph{\bibinfo{title}{Clifford group, stabilizer
  states, and linear and quadratic operations over GF(2)}}.
\newblock {\slshape \bibinfo{journal}{Phys. Rev. A}} \bibinfo{volume}{68}, p.
  \bibinfo{pages}{042318}, \doi{10.1103/PhysRevA.68.042318}.

\bibitemdeclare{article}{Duncan_2020}
\bibitem{Duncan_2020}
\bibinfo{author}{Ross \surnamestart Duncan\surnameend}, \bibinfo{author}{Aleks
  \surnamestart Kissinger\surnameend}, \bibinfo{author}{Simon \surnamestart
  Perdrix\surnameend} \& \bibinfo{author}{John \surnamestart van~de
  Wetering\surnameend} (\bibinfo{year}{2020}):
  \emph{\bibinfo{title}{Graph-theoretic {S}implification of {Q}uantum
  {C}ircuits with the {ZX}-calculus}}.
\newblock {\slshape \bibinfo{journal}{{Quantum}}} \bibinfo{volume}{4}, p.
  \bibinfo{pages}{279}, \doi{10.22331/q-2020-06-04-279}.

\bibitemdeclare{inproceedings}{Duncan_2009}
\bibitem{Duncan_2009}
\bibinfo{author}{Ross \surnamestart Duncan\surnameend} \&
  \bibinfo{author}{Simon \surnamestart Perdrix\surnameend}
  (\bibinfo{year}{2009}): \emph{\bibinfo{title}{Graph {States} and the
  {Necessity} of {Euler} {Decomposition}}}.
\newblock In \bibinfo{editor}{Klaus \surnamestart Ambos-Spies\surnameend},
  \bibinfo{editor}{Benedikt \surnamestart L{\"o}we\surnameend} \&
  \bibinfo{editor}{Wolfgang \surnamestart Merkle\surnameend}, editors:
  {\slshape \bibinfo{booktitle}{Mathematical Theory and Computational
  Practice}}, \bibinfo{publisher}{Springer Berlin Heidelberg},
  \bibinfo{address}{Berlin, Heidelberg}, pp. \bibinfo{pages}{167--177},
  \doi{10.1007/978-3-642-03073-4{\_}18}.

\bibitemdeclare{article}{Elliott_2008}
\bibitem{Elliott_2008}
\bibinfo{author}{Matthew~B. \surnamestart Elliott\surnameend},
  \bibinfo{author}{Bryan \surnamestart Eastin\surnameend} \&
  \bibinfo{author}{Carlton~M. \surnamestart Caves\surnameend}
  (\bibinfo{year}{2008}): \emph{\bibinfo{title}{Graphical description of the
  action of {Clifford} operators on stabilizer states}}.
\newblock {\slshape \bibinfo{journal}{Phys. Rev. A}} \bibinfo{volume}{77}, p.
  \bibinfo{pages}{042307}, \doi{10.1103/PhysRevA.77.042307}.

\bibitemdeclare{article}{Hu_2021}
\bibitem{Hu_2021}
\bibinfo{author}{Alexander~Tianlin \surnamestart Hu\surnameend} \&
  \bibinfo{author}{Andrey~Boris \surnamestart Khesin\surnameend}
  (\bibinfo{year}{2022}): \emph{\bibinfo{title}{Improved graph formalism for
  quantum circuit simulation}}.
\newblock {\slshape \bibinfo{journal}{Phys. Rev. A}} \bibinfo{volume}{105}, p.
  \bibinfo{pages}{022432}, \doi{10.1103/PhysRevA.105.022432}.

\bibitemdeclare{inproceedings}{Jeandel_2018}
\bibitem{Jeandel_2018}
\bibinfo{author}{Emmanuel \surnamestart Jeandel\surnameend},
  \bibinfo{author}{Simon \surnamestart Perdrix\surnameend} \&
  \bibinfo{author}{Renaud \surnamestart Vilmart\surnameend}
  (\bibinfo{year}{2018}): \emph{\bibinfo{title}{A {Complete} {Axiomatisation}
  of the {ZX}-{Calculus} for {Clifford+T} {Quantum} {Mechanics}}}.
\newblock In: {\slshape \bibinfo{booktitle}{Proceedings of the 33rd Annual
  ACM/IEEE Symposium on Logic in Computer Science}}, \bibinfo{series}{LICS
  '18}, \bibinfo{publisher}{Association for Computing Machinery},
  \bibinfo{address}{New York, NY, USA}, p. \bibinfo{pages}{559–568},
  \doi{10.1145/3209108.3209131}.

\bibitemdeclare{inproceedings}{Jeandel_2018_2}
\bibitem{Jeandel_2018_2}
\bibinfo{author}{Emmanuel \surnamestart Jeandel\surnameend},
  \bibinfo{author}{Simon \surnamestart Perdrix\surnameend} \&
  \bibinfo{author}{Renaud \surnamestart Vilmart\surnameend}
  (\bibinfo{year}{2018}): \emph{\bibinfo{title}{Diagrammatic {Reasoning} beyond
  {Clifford+T} {Quantum} {Mechanics}}}.
\newblock In: {\slshape \bibinfo{booktitle}{Proceedings of the 33rd Annual
  ACM/IEEE Symposium on Logic in Computer Science}}, \bibinfo{series}{LICS
  '18}, \bibinfo{publisher}{Association for Computing Machinery},
  \bibinfo{address}{New York, NY, USA}, p. \bibinfo{pages}{569–578},
  \doi{10.1145/3209108.3209139}.

\bibitemdeclare{article}{kissinger_reducing_2020}
\bibitem{kissinger_reducing_2020}
\bibinfo{author}{Aleks \surnamestart Kissinger\surnameend} \&
  \bibinfo{author}{John \surnamestart van~de Wetering\surnameend}
  (\bibinfo{year}{2020}): \emph{\bibinfo{title}{Reducing the number of
  non-{Clifford} gates in quantum circuits}}.
\newblock {\slshape \bibinfo{journal}{Physical Review A}}
  \bibinfo{volume}{102}(\bibinfo{number}{2}), p. \bibinfo{pages}{022406},
  \doi{10.1103/PhysRevA.102.022406}.

\bibitemdeclare{article}{miyazaki_analysis_2015}
\bibitem{miyazaki_analysis_2015}
\bibinfo{author}{Jisho \surnamestart Miyazaki\surnameend},
  \bibinfo{author}{Michal \surnamestart Hajdušek\surnameend} \&
  \bibinfo{author}{Mio \surnamestart Murao\surnameend} (\bibinfo{year}{2015}):
  \emph{\bibinfo{title}{Analysis of the trade-off between spatial and temporal
  resources for measurement-based quantum computation}}.
\newblock {\slshape \bibinfo{journal}{Physical Review A}}
  \bibinfo{volume}{91}(\bibinfo{number}{5}), p. \bibinfo{pages}{052302},
  \doi{10.1103/PhysRevA.91.052302}.

\bibitemdeclare{misc}{Ng_2017}
\bibitem{Ng_2017}
\bibinfo{author}{Kang~Feng \surnamestart Ng\surnameend} \&
  \bibinfo{author}{Quanlong \surnamestart Wang\surnameend}
  (\bibinfo{year}{2017}): \emph{\bibinfo{title}{A universal completion of the
  {ZX}-calculus}}, \doi{10.48550/arXiv.1706.09877}.

\bibitemdeclare{article}{Raussendorf_One-Way_2001}
\bibitem{Raussendorf_One-Way_2001}
\bibinfo{author}{Robert \surnamestart Raussendorf\surnameend} \&
  \bibinfo{author}{Hans~J. \surnamestart Briegel\surnameend}
  (\bibinfo{year}{2001}): \emph{\bibinfo{title}{A {One-Way} {Quantum}
  {Computer}}}.
\newblock {\slshape \bibinfo{journal}{Phys. Rev. Lett.}} \bibinfo{volume}{86},
  pp. \bibinfo{pages}{5188--5191}, \doi{10.1103/PhysRevLett.86.5188}.

\bibitemdeclare{inproceedings}{Simmons_2021}
\bibitem{Simmons_2021}
\bibinfo{author}{Will \surnamestart Simmons\surnameend} (\bibinfo{year}{2021}):
  \emph{\bibinfo{title}{Relating {Measurement} {Patterns} to {Circuits} via
  {Pauli} {Flow}}}.
\newblock In \bibinfo{editor}{Chris \surnamestart Heunen\surnameend} \&
  \bibinfo{editor}{Miriam \surnamestart Backens\surnameend}, editors: {\slshape
  \bibinfo{booktitle}{{\rm Proceedings 18th International Conference on}
  Quantum Physics and Logic, {\rm Gdansk, Poland, and online, 7-11 June
  2021}}}, {\slshape \bibinfo{series}{Electronic Proceedings in Theoretical
  Computer Science}} \bibinfo{volume}{343}, \bibinfo{publisher}{Open Publishing
  Association}, pp. \bibinfo{pages}{50--101}, \doi{10.4204/EPTCS.343.4}.

\bibitemdeclare{mastersthesis}{staudacher_optimization_2021}
\bibitem{staudacher_optimization_2021}
\bibinfo{author}{Korbinian \surnamestart Staudacher\surnameend}
  (\bibinfo{year}{2021}): \emph{\bibinfo{title}{Optimization {Approaches} for
  {Quantum} {Circuits} using {ZX}-calculus}}.
\newblock Master's thesis, \bibinfo{school}{Ludwig-Maximilians-Universität},
  \bibinfo{address}{München}.
\newblock
  \urlprefix\url{https://www.mnm-team.org/pub/Diplomarbeiten/stau21/PDF-Version/stau21.pdf}.

\bibitemdeclare{article}{VDN_2010}
\bibitem{VDN_2010}
\bibinfo{author}{Maarten \surnamestart Van Den~Nest\surnameend}
  (\bibinfo{year}{2010}): \emph{\bibinfo{title}{Classical {Simulation} of
  {Quantum} {Computation}, the {Gottesman-Knill} {Theorem}, and {Slightly}
  {Beyond}}}.
\newblock {\slshape \bibinfo{journal}{Quantum Info. Comput.}}
  \bibinfo{volume}{10}(\bibinfo{number}{3}), p. \bibinfo{pages}{258–271},
  \doi{10.5555/2011350.2011356}.

\bibitemdeclare{misc}{VDW_2020}
\bibitem{VDW_2020}
\bibinfo{author}{John \surnamestart van~de Wetering\surnameend}
  (\bibinfo{year}{2020}): \emph{\bibinfo{title}{ZX-calculus for the working
  quantum computer scientist}}, \doi{10.48550/ARXIV.2012.13966}.

\bibitemdeclare{misc}{wetering_2022}
\bibitem{wetering_2022}
\bibinfo{author}{John \surnamestart van~de Wetering\surnameend}
  (\bibinfo{year}{2022}): \emph{\bibinfo{title}{Personal communication}}.

\end{thebibliography}
